\documentclass{article}
\usepackage{graphicx} 

\usepackage{amsmath}
\usepackage{amsfonts}
\usepackage{amssymb}
\usepackage{amsthm}
\usepackage{siunitx}
\newif\ifshowrevs      
\showrevsfalse          
\usepackage{xcolor}
\usepackage[normalem]{ulem} 

\usepackage[colorlinks=true,citecolor=teal]{hyperref}
\usepackage[utf8]{inputenc}
\usepackage[section]{placeins}
\usepackage{fullpage}
\usepackage{booktabs}
\usepackage{cleveref}
\usepackage{verbatim}
\usepackage{xcolor}
\usepackage{natbib}
\usepackage{accents}

\usepackage{bm}
\usepackage{bbm}
\usepackage{float}      
\usepackage{subcaption} 
\theoremstyle{remark} \newtheorem*{remark}{Remark}

\usepackage[section]{algorithm}
\usepackage{algpseudocode}

\newtheorem{theorem}{Theorem}[section]
\newtheorem{lemma}[theorem]{Lemma}

\newtheorem{proposition}[theorem]{Proposition}
\newtheorem{corollary}[theorem]{Corollary}

\newtheorem{definition}[theorem]{Definition}

\newtheorem{assumption}[theorem]{Assumption}

\Crefname{problem}{Problem}{Problems}

\title{SPX--VIX Risk Computations Via Perturbed Optimal Transport}

\usepackage{authblk}
\author[1]{Charlie Che$^\ast$\thanks{charlie.che@jpmchase.com}}
\author[2]{Hanxuan Lin$^\ast$\thanks{hanxuan.lin@jpmchase.com}}
\author[1]{Yudong Yang$^\ast$\thanks{yudong.yang@jpmchase.com}}
\author[1]{Guofan Hu$^\ast$\thanks{guofan.hu@jpmchase.com}}
\author[1]{Lei Fang$^\ast$\thanks{lei.x.fang@jpmchase.com}}
\affil[1]{Quantitative Trading \& Research, JPMorganChase, New York, NY 10017, USA}
\affil[2]{Quantitative Research China, JPMorganChase, Beijing, 100033, China}

\newcommand{\margvec}{\mathbf{m}}
\newcommand{\vegaSPX}[1][]{\mathsf{V}^{BS}_{\mathrm{SPX}#1}}  
\newcommand{\vegaVIX}{\mathsf{V}^{BS}_{\mathrm{VIX}}}  
\begin{document}
\date{}
\maketitle
\begin{abstract}
We propose a model‑independent framework for joint SPX–VIX derivatives risk generation using Perturbed Optimal Transport (POT). The calibrated Gibbs coupling induces an exponential family whose Fisher information governs the response to admissible market shocks. Exploiting this structure, we compute sensitivities for payoffs written on the calibrated state space via a single Fisher-based linear solve, replacing bump‑and‑recalibrate procedures. To capture VIX smile dynamics, we introduce a linearized Skew Stickiness Ratio (SSR) rule to propagate SPX-induced VIX future shifts to the full VIX smile, completing the VIX marginal component $h_V$ of the perturbation vector. SSR determines how the VIX marginal target changes in response to an SPX shock, with second‑order error control. We also identify a conditional coupling invariance that reduces the perturbed transport on $(S_1, V, S_2)$ to an exact two‑dimensional projection on $(S_1, V)$, preserving martingality and variance consistency while lowering computational cost. Numerically, both the Fisher‑based linear‑response and the dimension‑reduced method closely match full recalibration for VIX futures and VIX option cross‑Greeks, yet are orders of magnitude faster. Hedging backtests show reduced hedged P\&L variance versus a stochastic local‑volatility benchmark, especially in volatile regimes. We further demonstrate that the framework extends beyond VIX vanilla options: a VIX-barrier forward-starting SPX call is priced and risked using the same Fisher linear solve, with the linear-response delta capturing a non-trivial sensitivity arising from variance-consistency recalibration even for a scale-invariant payoff. These results establish POT, coupled with linear response and SSR-based VIX smile propagation, as a practical, efficient framework for SPX–VIX risk propagation and hedging.
\end{abstract}


\newpage
\section*{Notation}
\vspace{-2pt}

\noindent
\begin{minipage}[t]{0.48\textwidth}
\footnotesize\renewcommand{\arraystretch}{1.13}
\begin{tabular}{@{}p{0.44\linewidth}p{0.52\linewidth}@{}}
\toprule
\textbf{Symbol} & \textbf{Description}\\
\midrule
\multicolumn{2}{@{}l}{\textit{State spaces and dimension}}\\[1pt]
$\mathcal{S}_1=\{s_1^1,\dots,s_1^{N_1}\}$ & SPX grid at $T_1$\\
$\mathcal{S}_2=\{s_2^1,\dots,s_2^{N_2}\}$ & SPX grid at $T_2$\\
$\mathcal{V}=\{v^1,\dots,v^{N_V}\}$ & VIX grid\\
$\mathcal{X}=\mathcal{S}_1\!\times\!\mathcal{V}\!\times\!\mathcal{S}_2$ & Full product state space\\
$\Delta(\mathcal{X})$ & Probability simplex on $\mathcal{X}$\\
$K=N_1\!+\!N_V\!+\!N_2\!+\!2N_1N_V$ & Total parameter/statistic dimension\\[3pt]
\multicolumn{2}{@{}l}{\textit{Time horizons}}\\[1pt]
$T_1$ & First expiry (SPX and VIX)\\
$T_2$ & Second SPX expiry\\
$\tau=T_2-T_1$ & Time interval ($\tau=30$ days)\\[3pt]
\multicolumn{2}{@{}l}{\textit{Probability measures and marginals}}\\[1pt]
$D(\mu\|\bar\mu)=\sum_x\mu\ln(\mu/\bar\mu)$ & KL divergence from $\bar\mu$ (minimized in MOT)\\
$\bar{\mu}(x)>0$ & Prior coupling on $\mathcal{X}$\\
$\mu_1,\mu_V,\mu_2$ & Target marginals for $S_1$, $V$, $S_2$\\
$\margvec=(\mu_1,\mu_V,\mu_2)$ & Stacked marginal target vector\\
$\mu^\star$ & Calibrated MOT optimal coupling\\
$\mu^\varepsilon$ & Perturbed coupling (shock size $\varepsilon$)\\
$\gamma^\star(s_1,v):=\sum_{s_2}\mu^\star(s_1,v,s_2)$ & $(S_1,V)$-marginal of $\mu^\star$\\
$\gamma^\varepsilon$ & Perturbed reduced coupling on $\mathcal{S}_1\!\times\!\mathcal{V}$ (DR)\\
$\mu_2^{DR}=\sum_{s_1,v}\kappa^\star\gamma^\varepsilon$ & Endogenous $S_2$ marginal under DR\\
$\kappa^\star(s_2\mid s_1,v)$ & Conditional kernel; fixed under Assumption~\ref{ass:conditional_invariance}\\[3pt]
\multicolumn{2}{@{}l}{\textit{Constraint operators and admissible sets}}\\[1pt]
$\mathcal{A}_{\mathrm{marg}}$ & Marginal constraint operator\\
$\mathcal{A}_{\mathrm{fin}}$ & Financial constraint operator\\
$\mathcal{A}=(\mathcal{A}_{\mathrm{marg}};\mathcal{A}_{\mathrm{fin}})$ & Full stacked constraint operator\\
$\mathcal{P}$ & 3D MOT-admissible coupling set\\
$\mathcal{P}^\varepsilon$ & Perturbed 3D admissible set (LR)\\
$\mathcal{Q}^\varepsilon$ & Perturbed 2D admissible set (DR); see~\eqref{eq:Q_eps}\\
$b_k$ & Constraint RHS ($b_k=(\margvec)_k$ for marginal slots, $b_k=0$ for financial slots)\\
$h=(h_1,h_V,h_2,\mathbf{0})\in\mathbb{R}^K$ & Full perturbation vector (financial slots zero); $\dot b=h$ (synonym, see~\eqref{eq:perturbation_vector})\\
$h_{\mathrm{marg}}:=(h_1,h_V,h_2)$ & Marginal block of $h$\\
\bottomrule
\end{tabular}
\end{minipage}%
\hfill
\begin{minipage}[t]{0.48\textwidth}
\footnotesize\renewcommand{\arraystretch}{1.13}
\begin{tabular}{@{}p{0.44\linewidth}p{0.52\linewidth}@{}}
\toprule
\textbf{Symbol} & \textbf{Description}\\
\midrule
\multicolumn{2}{@{}l}{\textit{Dual parameters}}\\[1pt]
$\theta=(u_1,u_V,u_2,\{\Delta_M^{ij}\},\{\Delta_C^{ij}\})$ & Full gauge-fixed parameter vector\\
$u_1,\;u_V,\;u_2$ & Marginal dual potentials\\
$a\!=\!e^{u_1}$,\;$q\!=\!e^{u_V}$,\;$c\!=\!e^{u_2}$ & Sinkhorn scaling factors (Algorithm~1)\\
$q(v)=e^{u_V(v)}$ & VIX Sinkhorn scaling factor ($q$ replaces earlier $b$ to avoid clash with constraint RHS $b_k$)\\
$\Delta_M(s_1,v)$ & Delta-hedge Lagrange multiplier\\
$\Delta_C(s_1,v)$ & Vega-hedge Lagrange multiplier\\
$\dot{\theta}^0=d\theta^\varepsilon/d\varepsilon|_{\varepsilon=0}$ & First-order variation of $\theta^\star$\\[3pt]
\multicolumn{2}{@{}l}{\textit{Sufficient statistics, log-partition, Fisher information}}\\[1pt]
$\{T_k\}_{k=1}^K$ & Sufficient statistics (see \eqref{eq:MOT_stats})\\
$T^{\mathrm{marg}},\;T^{\Delta_M},\;T^{\Delta_C}$ & Marginal, delta, vega stat blocks\\
$\Lambda(\theta)=\ln\!\sum_x\bar\mu\,e^{\langle\theta,T\rangle}$ & Log-partition function\\
$H=\nabla^2_\theta\Lambda(\theta^\star)\succeq 0$ & Fisher information matrix; $H\succ 0$ under Assumption~\ref{ass:non_deg}\\
$H^+$ & Moore--Penrose pseudoinverse; $H^+=H^{-1}$ under Assumption~\ref{ass:non_deg}\\
$H_{\mathrm{marg}},H_{MM},H_{CC},H_{MC}$ & Diagonal blocks of $H$ (see \eqref{eq:MOT_fisher_block})\\
$H_{\mathrm{marg},M},H_{\mathrm{marg},C}$ & Off-diagonal coupling blocks of $H$\\[3pt]
\multicolumn{2}{@{}l}{\textit{Payoffs, perturbations, and risk}}\\[1pt]
$L(x)=-\tfrac{2}{\tau}\ln x$ & Log-return (VIX--variance link)\\
$\sigma^2_{\mathrm{fwd}}=\mathbb{E}_{\mu^\star}[L(S_2/S_1)]$ & SPX-implied forward variance\\
$G:\mathcal{X}\to\mathbb{R}$ & Payoff function\\
$\Pi(\varepsilon)=\mathbb{E}_{\mu^\varepsilon}[G]$ & Model price under $\mu^\varepsilon$\\
$(g_G)_k=\mathrm{Cov}_{\mu^\star}(G,T_k)$ & Payoff--statistic covariance; $g_G\in\mathrm{range}(H)$\\
$\Psi=H^+g_G$ & Full influence vector\\
$\psi=(\psi_1,\psi_V,\psi_2)$ & Marginal influence functions\\[3pt]
\multicolumn{2}{@{}l}{\textit{SSR and VIX propagation}}\\[1pt]
$\delta_S$ & SPX spot bump size (§~\ref{subsec:spot_perturbation})\\
$\delta_\sigma$ & SPX parallel vol bump size (§~\ref{subsec:vol_perturbation})\\
$\delta F_V$ & First-order VIX future shift\\
$A_{\mathrm{SSR}}(v)=\partial^2_K\Gamma|_{K=v}$ & SSR shape vector; $h_V=A_\mathrm{SSR}\delta F_V$\\
$\Gamma(K)$ & SSR sensitivity kernel (see \eqref{eq:deltaCV})\\
$\vegaSPX[,i](K):=\partial_\sigma C_{BS}(S_0,K,T_i,\sigma)$ & SPX Black-Scholes option vega at maturity $T_i$; see \eqref{eq:hi_vol}\\
$\vegaVIX(K):=\partial_{\sigma_V}C_V(K)$ & VIX option Black-Scholes vega\\
$\mathsf{V}_{\mathrm{SPX}}$ & SPX vega (model-implied price sensitivity to SPX vol)\\[3pt]
\multicolumn{2}{@{}l}{\textit{Skew Stickiness Ratio (SSR)}}\\[1pt]
$F,\;F_V$ & SPX and VIX forward levels\\
$\sigma(K,F)$,\;$\sigma_V(K,F_V)$ & SPX and VIX implied volatilities\\
$S_{\log}=F\,\partial_K\sigma|_{K=F}$ & SPX ATM log-moneyness skew\\
$S_{\log,V}=F_V\,\partial_K\sigma_V|_{K=F_V}$ & VIX ATM log-moneyness skew\\
$\beta$,\;$\beta_V$ & SPX and VIX Skew Stickiness Ratios\\
\bottomrule
\end{tabular}
\end{minipage}

\newpage
\section{Introduction}

Joint modeling of SPX and VIX derivatives has become a central problem in equity volatility markets. 
While SPX options encode the distribution of future equity prices, VIX options are derivative contracts written on VIX, the square root of forward variance. 
Consistency between these two markets is therefore essential for both pricing and risk management of the SPX/VIX derivative family.

Traditional approaches to the joint SPX–VIX modeling problem rely on parametric stochastic volatility models such as the Heston model, stochastic volatility with jumps, Bergomi model, or rough volatility frameworks \cite{heston1993, gatheral2006, bergomi2016, BayerFrizRoughVol}.
Although these models provide tractable simulation dynamics, they impose structural assumptions on volatility dynamics that are not directly implied by the observed option surfaces.
More fundamentally, traditional parametric continuous-time model struggle to simultaneously recover both the SPX and VIX marginals; practitioners therefore typically decouple the two surfaces, which means that cross-sensitivities between the SPX volatility surface and VIX derivatives are not captured at all.

More recently, several alternative approaches have been developed to address the joint SPX-VIX calibration problem beyond classical parametric stochastic-volatility models.  Cuchiero et al.\cite{cuchiero2025joint} propose a flexible signature based framework for joint SPX-VIX modeling and calibration, while Abi Jaber et al.\cite{abijaber2023quintic} introduce the quintic Ornstein-Uhlenbeck volatility model, which achieves acurate joint calibration through a parsimonious polynomial transformation of an Ornstein-Uhlenbeck factor.  Moving away from prescribed parametric volatility dynamics, Guo et al.~\cite{guo2022joint} formulate joint SPX-VIX calibration through semimartingale optimal transport, constructing calibrated continuous-time models subject to the relevant market constraints.  These developments illustrate a broader shift towards more flexible representations and optimization-based formulations of the joint calibration problems.

Partiuclarly relevant to the present work is the nonparametric approach of Guyon in \cite{guyon2020spxvix, Guyon_SSRN_3853237}, who formulated the joint SPX–VIX calibration problem as a discrete martingale optimal transport (MOT) problem. 
In this framework the calibrated coupling between equity levels and forward variance is obtained by solving a discrete entropic optimal transport problem using Sinkhorn iterations \cite{cuturi2013, benamou2015}. 
The resulting Gibbs distribution exactly reproduces the observed SPX and VIX option prices while remaining free of parametric volatility assumptions.

While entropic martingale optimal transport provides an exact joint
calibration of SPX and VIX smiles, calibration alone does not yield a
practical risk framework. In existing implementations, sensitivities
are typically obtained by re-running the full calibration after each
market perturbation. This bump-and-recalibrate approach is both
computationally expensive and obscures the structural relationship
between market shocks and model-implied risk.

The central observation of this paper is that entropic martingale optimal transport naturally defines a statistical manifold whose local geometry determines how the calibrated coupling reacts to marginal perturbations.
Because the optimal coupling belongs to an exponential family, its response to marginal shocks can be characterized by the Fisher information matrix of the calibrated Gibbs distribution.
This perspective leads to a new framework, which we term \emph{Perturbed Optimal Transport (POT)}, for risk generation without recalibration. Within this POT framework, we propose two distinct yet complementary methodologies: one leveraging the local geometry of the calibrated coupling via a Linear Response (LR) system derived from the Fisher information matrix, and another utilizing Dimensional Reduction (DR) to efficiently re-solve a simpler transport problem under specific conditional invariance assumptions.

Beyond providing analytic risk formulas, the POT framework also enables incorporating empirical volatility dynamics into the optimal transport formulation. 
In particular, we use the Skew Stickiness Ratio (SSR) dynamics for the VIX volatility surface to determine the VIX marginal component of the perturbation vector: an SPX perturbation shifts the VIX future via the forward-variance identity, and the linearized SSR rule then propagates this shift to the full VIX smile, specifying how the VIX marginal target $\mu_V$ changes.
The constraint operator and dual structure are unchanged; only the right-hand side of the linear system, the perturbation vector $\dot b$, is updated.
This allows the transport framework to incorporate empirically observed volatility smile dynamics without introducing parametric stochastic volatility models.

Taken together, these results establish Perturbed Optimal Transport (POT) as a unified framework for joint calibration and risk propagation in SPX--VIX markets, offering efficient risk generation through both Linear Response and Dimensional Reduction techniques.
\subsection{Related Literature}

This work relates to two strands of literature.

\textbf{Optimal transport in finance.}

Martingale optimal transport has emerged as a model-free approach to derivative pricing and calibration \cite{beiglbock2013, henrylabordere2017}. 
\cite{guyon2020spxvix,Guyon_SSRN_3853237} introduced an entropic optimal transport formulation for the joint calibration of SPX and VIX smiles, which can be solved efficiently using Sinkhorn iterations.

\textbf{Computational optimal transport and entropy regularization.}

Entropy-regularized transport problems have become widely used in machine learning and computational optimal transport due to their favorable numerical properties \cite{cuturi2013, benamou2015, peyre_cuturi_2019}. 
These formulations lead to Gibbs distributions whose structure enables efficient iterative algorithms.

Our contributions extend this literature by showing that entropic MOT calibration naturally induces a perturbation theory that can be used to generate risk sensitivities without recomputing the transport solution.

\subsection{Main Contributions}

The contributions of this paper are as follows.

\begin{enumerate}

\item \textbf{The Linear Response (LR) System For Perturbed Optimal Transport In SPX--VIX Markets.}

We develop a perturbation framework for discrete entropic optimal transport
under both marginal and financial constraints.
Using the implicit function theorem applied to the dual formulation,
we show that the calibrated Gibbs coupling depends smoothly on admissible
market perturbations.
This yields explicit linear-response formulas for sensitivities of arbitrary
payoffs, governed by the Fisher information matrix of the calibrated
exponential family.

\item \textbf{Linearized Skew Stickiness Ratio dynamics for VIX options.}

We introduce a linearization of Skew Stickiness Ratio (SSR) dynamics for
VIX implied volatility surfaces and use it to determine the VIX marginal
component $h_V$ of the perturbation vector: the SSR rule maps the
SPX-induced VIX future shift to the full VIX smile change, updating the
right-hand side $\dot b$ of the linear system while leaving the constraint
operator unchanged.
This formulation provides a model-independent mechanism for propagating
SPX perturbations to the VIX volatility smile while preserving convexity
and tractability of the entropic projection problem.
The SSR linearization is compatible with both the perturbation-based
linear-response risk engine and the dimension-reduced transport framework
developed later in the paper.

\item \textbf{Dimensional Reduction (DR) Within The Perturbed Optimal Transport Framework.}

We identify a conditional coupling invariance structure under which
the perturbed three-dimensional transport problem on
$(S_1,V,S_2)$ reduces to a two-dimensional entropic projection
on $(S_1,V)$.
Under this structure the conditional kernel of $S_2$ given $(S_1,V)$
remains fixed, so martingality and variance-consistency constraints
are automatically preserved.
This reduction dramatically lowers the computational complexity of
risk generation while maintaining financial consistency of the model.

\item \textbf{Empirical validation and hedging backtests.}

We compare the LR and DR risk sensitivities against full recalibration of the
SPX--VIX martingale optimal transport model across VIX futures and VIX option
cross-Greeks, finding close agreement at orders-of-magnitude lower computational
cost.
We further conduct hedging backtests on randomized VIX option portfolios;
the transport-based hedges consistently achieve lower hedged P\&L variance than
a stochastic local volatility (SLV) benchmark, particularly during volatile
market regimes.
As a further demonstration, we price and risk a VIX-barrier forward-starting SPX call whose payoff depends jointly on $V$ and $S_2/S_1$. Unlike VIX vanilla claims, this payoff cannot in general be valued from the calibrated marginals alone and therefore genuinely requires the joint SPX-VIX coupling. It provides a setting in which joint calibration and consequently perturbation of the joint transport plan is essential rather than merely a convenient route to cross-risk.

\end{enumerate}

These contributions establish Perturbed Optimal Transport (POT) as a unified
framework for SPX--VIX joint calibration, risk generation, and hedging.
Throughout, the calibrated coupling is required to satisfy a martingale
condition $\mathbb{E}[S_2\mid S_1,V]=S_1$ and a variance-consistency condition
$\mathbb{E}[L(S_2/S_1)\mid S_1,V]=V^2$; these structural conditions are defined
formally in Section~\ref{sec:framework}.

\section{Mathematical Framework For Entropic Martingale Optimal Transport}
\label{sec:framework}

\subsection{Setup: State Space And Admissible Couplings}
\label{subsec:setup}

Let
\[
\mathcal{S}_1 = \{s_1^1,\dots,s_1^{N_1}\},\quad
\mathcal{S}_2 = \{s_2^1,\dots,s_2^{N_2}\},\quad
\mathcal{V} = \{v^1,\dots,v^{N_V}\}
\]
be finite state spaces for the SPX level at $T_1$, SPX level at $T_2$, and
VIX level.  Let $\mathcal{X}:=\mathcal{S}_1\times\mathcal{V}\times\mathcal{S}_2$,
and let $\bar\mu(x)>0$ be a strictly positive prior on $\mathcal{X}$.
The target marginals $\mu_1\in\Delta(\mathcal{S}_1)$,
$\mu_V\in\Delta(\mathcal{V})$, $\mu_2\in\Delta(\mathcal{S}_2)$
are extracted from observed SPX and VIX option prices.

In the SPX--VIX joint calibration problem, an admissible coupling must
satisfy two kinds of linear constraints:
\begin{itemize}
  \item \textbf{Marginal constraints} ($\mathcal{A}_{\mathrm{marg}}$):
        $\mu$ reproduces the three prescribed marginals.
  \item \textbf{Financial constraints} ($\mathcal{A}_{\mathrm{fin}}$):
        martingality of SPX and consistency between VIX and forward variance
        (Definition~\ref{def:constraint_operators}).
\end{itemize}
Both sets are linear in $\mu$, so the \emph{MOT-admissible set} is
\begin{equation}
\label{eq:MOT_admissible}
\mathcal{P}
=
\bigl\{\,\mu \in \Delta(\mathcal{X})
:
\mathcal{A}_{\mathrm{marg}}\mu = \margvec,\;\;
\mathcal{A}_{\mathrm{fin}}\mu = 0
\bigr\},
\end{equation}
where $\margvec = (\mu_1,\mu_V,\mu_2)$ stacks the marginal targets.

\begin{definition}[Constraint operators: marginalisation and financial consistency]
\label{def:constraint_operators}
The linear operator
$\mathcal{A}_{\mathrm{marg}}:\mathbb{R}^{\mathcal{X}}\to\mathbb{R}^{N_1+N_V+N_2}$
marginalizes $\mu$ onto each coordinate axis:
\begin{align*}
(\mathcal{A}_{\mathrm{marg}}\mu)_1(s_1^i)
  &= \sum_{v,\,s_2}\mu(s_1^i,v,s_2), \\
(\mathcal{A}_{\mathrm{marg}}\mu)_V(v^j)
  &= \sum_{s_1,\,s_2}\mu(s_1,v^j,s_2), \\
(\mathcal{A}_{\mathrm{marg}}\mu)_2(s_2^k)
  &= \sum_{s_1,\,v}\mu(s_1,v,s_2^k).
\end{align*}
The linear operator
$\mathcal{A}_{\mathrm{fin}}:\mathbb{R}^{\mathcal{X}}\to\mathbb{R}^{2N_1N_V}$
evaluates the conditional moment residuals at each node $(s_1^i,v^j)$:
\begin{align*}
(\mathcal{A}_{\mathrm{fin}}\mu)^M_{ij}
  &= \sum_{s_2}\mu(s_1^i,v^j,s_2)\,(s_2 - s_1^i),
  &&\text{\emph{(martingale block)}}\\
(\mathcal{A}_{\mathrm{fin}}\mu)^C_{ij}
  &= \sum_{s_2}\mu(s_1^i,v^j,s_2)\,\bigl(L(s_2/s_1^i)-(v^j)^2\bigr).
  &&\text{\emph{(variance-consistency block)}}
\end{align*}
Writing $\mathcal{A}:=(\mathcal{A}_{\mathrm{marg}};\mathcal{A}_{\mathrm{fin}})$
and $b:=(\margvec;\,0)\in\mathbb{R}^{N_1+N_V+N_2+2N_1N_V}$,
the admissible set~\eqref{eq:MOT_admissible} reads
$\mathcal{P}=\{\mu\in\Delta(\mathcal{X}):\mathcal{A}\mu=b\}$.
\end{definition}

\begin{remark}[Financial interpretation of $\mathcal{A}_{\mathrm{fin}}\mu=0$]
\label{rem:fin_interpretation}
Setting $\mathcal{A}_{\mathrm{fin}}\mu=0$ imposes two families of conditions on the
calibrated coupling, one per conditioning node $(s_1^i,v^j)$:
\begin{enumerate}
\item \textbf{Martingale condition}
      ($(\mathcal{A}_{\mathrm{fin}}\mu)^M_{ij}=0$):
      \[
      \mathbb{E}_\mu[S_2 \mid S_1=s_1^i,\,V=v^j] \;=\; s_1^i,
      \]
      i.e.\ the SPX price is a risk-neutral martingale conditional on
      the SPX level and VIX index at $T_1$.
      We assume zero interest rates, dividends, and repos throughout,
      following \cite{guyon2020spxvix}; with nonzero rates the condition
      becomes $\mathbb{E}[S_2/F_2\mid S_1,V]=S_1/F_1$ where $F_i$ is the
      forward price at $T_i$, and the grid is interpreted as normalized
      by the respective forwards.

\item \textbf{Variance-consistency condition}
      ($(\mathcal{A}_{\mathrm{fin}}\mu)^C_{ij}=0$):
      \begin{equation}
      \label{eq:consistency_condition}
      \mathbb{E}_\mu\!\bigl[L(S_2/S_1)\mid S_1=s_1^i,\,V=v^j\bigr]
      \;=\; (v^j)^2,
      \end{equation}
      i.e.\ the conditional expected log-contract payoff equals the
      VIX squared at $T_1$.  This is the discrete form of Guyon's
      consistency condition \cite{guyon2020spxvix}: by definition the
      VIX is the square root of the forward-realized variance of the
      SPX, so this identity must hold pointwise on the $(S_1,V)$ grid.
\end{enumerate}
Both conditions are linear in $\mu$ with zero right-hand side, confirming
that the financial constraint targets in $b$ are zero.
\end{remark}

\subsection{Entropic MOT: Primal--Dual Problem and Gibbs Form}
\label{subsec:EOT}

We minimize relative entropy over the MOT-admissible set:
\begin{equation}
\label{eq:primal}
\mu^\star
=
\arg\inf_{\mu \in \mathcal{P}}
D(\mu \| \bar{\mu}),
\qquad
D(\mu\|\bar\mu) = \sum_{x\in\mathcal{X}}\mu(x)\ln\frac{\mu(x)}{\bar\mu(x)}.
\end{equation}

\paragraph{Sufficient statistics and log-partition function.}
The constraint operator $\mathcal{A}$ of Definition~\ref{def:constraint_operators}
and prior $\bar\mu$ determine the following sufficient statistics and log-partition
function, which parametrize the exponential family of Gibbs couplings:
\begin{equation}
\label{eq:MOT_stats}
\{T_k\}
  =\underbrace{\{\mathbf{1}_{s_1=s_1^i},\;\mathbf{1}_{v=v^j},\;
                 \mathbf{1}_{s_2=s_2^k}\}}_{\text{marginal indicators}\;T^{\mathrm{marg}}}
  \;\cup\;
  \underbrace{\bigl\{\mathbf{1}_{(s_1,v)=(s_1^i,v^j)}\,(s_2-s_1^i)
              \bigr\}_{i,j}}_{\text{delta payoffs}\;T^{\Delta_M}}
  \;\cup\;
  \underbrace{\bigl\{\mathbf{1}_{(s_1,v)=(s_1^i,v^j)}\,
              \bigl(L(s_2/s_1^i)-(v^j)^2\bigr)
              \bigr\}_{i,j}}_{\text{vega payoffs}\;T^{\Delta_C}},
\end{equation}
with full Lagrange multiplier vector
$\theta = \bigl(u_1,u_V,u_2,\{\Delta_M^{ij}\},\{\Delta_C^{ij}\}\bigr)\in\mathbb{R}^K$,
where $\Delta_M^{ij}:=\Delta_M(s_1^i,v^j)$ and $\Delta_C^{ij}:=\Delta_C(s_1^i,v^j)$
are the node-indexed shorthand for the financial multipliers
and log-partition function
\begin{equation}
\label{eq:logpartition}
\Lambda(\theta)
= \ln\sum_{x\in\mathcal{X}}\bar\mu(x)\,e^{\langle\theta,\,T(x)\rangle}.
\end{equation}
The two objects $T$ and $\theta$ play distinct roles: $\{T_k\}$ are
functions on $\mathcal{X}$: the row functions of $\mathcal{A}$, hence the
generating set of its adjoint $\mathcal{A}^*$.
The vector $\theta\in\mathbb{R}^K$ is the Lagrange multiplier for the
constraint $\mathcal{A}\mu=b$; it lives in the codomain of $\mathcal{A}$
and is not a function on $\mathcal{X}$.
The Gibbs exponent $(\mathcal{A}^*\theta)(x)=\langle\theta,T(x)\rangle$
is the pairing of these two objects.

\begin{proposition}[Sufficient statistics as the adjoint of the constraint operator]
\label{prop:adjoint}
Equip $\mathbb{R}^{\mathcal{X}}$ with the inner product
$\langle f,g\rangle_{\mathcal{X}}:=\sum_{x\in\mathcal{X}}f(x)g(x)$
and $\mathbb{R}^{N_1+N_V+N_2+2N_1N_V}$ with the standard Euclidean inner product.
Let $\mathcal{A}$ be the operator of Definition~\ref{def:constraint_operators}
and $\{T_k\}$ the sufficient statistics of~\eqref{eq:MOT_stats}.
\begin{enumerate}
\item[(i)] $(\mathcal{A}\mu)_k = \langle T_k,\,\mu\rangle_{\mathcal{X}}$
           for every index $k$.
\item[(ii)] The adjoint $\mathcal{A}^*:\mathbb{R}^{N_1+N_V+N_2+2N_1N_V}\to\mathbb{R}^{\mathcal{X}}$
            satisfies
            \[
            (\mathcal{A}^*\theta)(x)
            = \langle\theta,\,T(x)\rangle
            = \sum_k \theta_k\,T_k(x).
            \]
\item[(iii)] The KKT stationarity condition for~\eqref{eq:primal} gives
             $\ln(\mu^\star/\bar\mu) = \mathcal{A}^*\theta^\star$
             (up to an additive constant absorbed into the marginal-potential
             gauge), recovering the Gibbs form of the calibrated coupling.
\end{enumerate}
\end{proposition}
\begin{proof}
\textit{(i)} By Definition~\ref{def:constraint_operators}, each entry of
$\mathcal{A}\mu$ is a sum over $x\in\mathcal{X}$ of $\mu(x)$ weighted by
exactly one element of $\{T_k\}$.  For instance, for the marginal index
$k=(1,i)$:
\[
(\mathcal{A}_{\mathrm{marg}}\mu)_1(s_1^i)
= \sum_{v,s_2}\mu(s_1^i,v,s_2)
= \sum_{x\in\mathcal{X}}\mathbf{1}_{s_1(x)=s_1^i}\,\mu(x)
= \langle T_{1,i},\,\mu\rangle_{\mathcal{X}}.
\]
The same holds for every other index by inspection, so
$(\mathcal{A}\mu)_k = \langle T_k,\mu\rangle_{\mathcal{X}}$ for all $k$.

\medskip\noindent
\textit{(ii)} By definition of the adjoint,
$\langle\theta,\mathcal{A}\mu\rangle = \langle\mathcal{A}^*\theta,\mu\rangle_{\mathcal{X}}$
for all $\mu$ and $\theta$.  Using (i):
\[
\langle\theta,\mathcal{A}\mu\rangle
= \sum_k \theta_k\,(\mathcal{A}\mu)_k
= \sum_k \theta_k\,\langle T_k,\mu\rangle_{\mathcal{X}}
= \Bigl\langle\sum_k \theta_k T_k,\;\mu\Bigr\rangle_{\mathcal{X}}.
\]
Since this holds for all $\mu$, we identify
$(\mathcal{A}^*\theta)(x) = \sum_k\theta_k T_k(x) = \langle\theta,T(x)\rangle$.

\medskip\noindent
\textit{(iii)} Stationarity of
$\mathcal{L}(\mu,\theta)=D(\mu\|\bar\mu)-\langle\theta,\mathcal{A}\mu-b\rangle$
at $\mu^\star$ gives, for each $x\in\mathcal{X}$:
\[
\frac{\partial\mathcal{L}}{\partial\mu(x)}\bigg|_{\mu^\star}
= \ln\frac{\mu^\star(x)}{\bar\mu(x)} + 1 - (\mathcal{A}^*\theta^\star)(x) = 0,
\]
hence $\ln(\mu^\star(x)/\bar\mu(x)) = \langle\theta^\star,T(x)\rangle - 1$.
The $+1$ arising from differentiating $D(\mu\|\bar\mu)$ produces the
constant $-1$ on the right-hand side above.
Since $\sum_k T_{1,k}(x)\equiv 1$ on $\mathcal{X}$ (exactly one
$S_1$-indicator fires at every $x$), the constant function $-1$ lies in
$\mathrm{range}(\mathcal{A}^*)$: it can be written as
$\langle c,T(x)\rangle$ by setting $c_{1,k}=-1$ for all $k$ and
$c_j=0$ otherwise.
It is therefore absorbed into the marginal-potential component of
$\theta^\star$, giving
$\mu^\star(x)\propto\bar\mu(x)\,e^{\langle\theta^\star,T(x)\rangle}$.
\end{proof}

\begin{theorem}[Existence, Uniqueness, and Gibbs Form]
\label{thm:existence}
Assume:
\begin{enumerate}
  \item[(A1)] $\bar\mu(x)>0$ for all $x\in\mathcal{X}$;
  \item[(A2)] $\mathcal{P}\neq\emptyset$;
  \item[(A3)] there exists $\mu_0\in\mathcal{P}$ with $\mu_0(x)>0$ for all
              $x\in\mathcal{X}$ \emph{(strictly positive feasible coupling)}.
\end{enumerate}
Then problem~\eqref{eq:primal} has a unique minimizer $\mu^\star$ with
$\mu^\star(x)>0$ for all $x\in\mathcal{X}$, and it
takes the \emph{extended Gibbs form}
\begin{equation}
\label{eq:MOT_gibbs}
\mu^\star(s_1,v,s_2)\;\propto\;\bar\mu(s_1,v,s_2)\,\exp\!\Bigl(
  \underbrace{u_1(s_1)+u_V(v)+u_2(s_2)}_{\text{marginal potentials}}
  +\underbrace{\Delta_M(s_1,v)\,(s_2-s_1)}_{\text{delta hedge}}
  +\underbrace{\Delta_C(s_1,v)\,\bigl(L(s_2/s_1)-v^2\bigr)}_{\text{vega hedge}}
\Bigr),
\end{equation}
where $(u_1,u_V,u_2)$ are dual potentials for the marginal constraints
and $\{\Delta_M(s_1,v),\,\Delta_C(s_1,v)\}$ are Lagrange multipliers for
the financial constraints.
\end{theorem}
\begin{proof}
We follow the argument of \citeauthor{guyon2020spxvix}~\cite{guyon2020spxvix},
Section~5, making each step explicit.

\paragraph{Existence and uniqueness.}
Since $\mathcal{X}$ is finite, $\Delta(\mathcal{X})$ is compact.
With the convention $0\ln 0:=0$, the map $t\mapsto t\ln t$ is continuous on
$[0,\infty)$, and $\bar\mu(x)>0$ for every $x$ ensures that
$\mu\mapsto-\sum_x\mu(x)\ln\bar\mu(x)$ is continuous on $\Delta(\mathcal{X})$.
Hence $D(\cdot\|\bar\mu)$ is continuous on $\Delta(\mathcal{X})$.
Strict convexity of $t\mapsto t\ln t$ on $[0,\infty)$ makes
$\mu\mapsto\sum_x\mu(x)\ln\mu(x)$ strictly convex; adding the linear term
$-\sum_x\mu(x)\ln\bar\mu(x)$ preserves strict convexity.

The feasible set $\mathcal{P}=\{\mu\in\Delta(\mathcal{X}):\mathcal{A}\mu=b\}$
is the intersection of the simplex with an affine subspace, hence closed, convex,
and compact; by hypothesis it is non-empty.
Weierstrass's theorem gives a minimizer $\mu^\star$.
If $\tilde\mu_1,\tilde\mu_2\in\mathcal{P}$ were both minimizers with $\tilde\mu_1\neq\tilde\mu_2$,
then $\tfrac{1}{2}(\tilde\mu_1+\tilde\mu_2)\in\mathcal{P}$ by convexity, and strict
convexity of $t\mapsto t\ln t$ gives
\[
D(\tfrac{1}{2}(\tilde\mu_1+\tilde\mu_2)\,\|\,\bar\mu)
< \tfrac{1}{2}D(\tilde\mu_1\,\|\,\bar\mu)+\tfrac{1}{2}D(\tilde\mu_2\,\|\,\bar\mu)
= \min_{\mathcal{P}} D(\cdot\,\|\,\bar\mu).
\]
This contradicts minimality of $\tilde\mu_1$ and $\tilde\mu_2$. Hence $\mu^\star$ is unique.

\paragraph{Strict positivity of $\mu^\star$.}
We show $\mu^\star(x)>0$ for all $x$.
Suppose for contradiction that $S^c:=\{x\in\mathcal{X}:\mu^\star(x)=0\}$ is
nonempty.
By (A3) there exists $\mu_0\in\mathcal{P}$ with $\mu_0>0$ everywhere.
For $t\in(0,1]$, set $\mu_t:=(1-t)\mu^\star+t\mu_0\in\mathcal{P}$ (by
convexity of $\mathcal{P}$).
Decompose $D(\mu_t\|\bar\mu) - D(\mu^\star\|\bar\mu)$ into contributions
from $S:=\mathcal{X}\setminus S^c$ (where $\mu^\star>0$, so $F(t):=\sum_{x\in S}\mu_t(x)\ln(\mu_t(x)/\bar\mu(x))$ is $C^1$ near $t=0$) and from $S^c$:
\[
D(\mu_t\|\bar\mu) - D(\mu^\star\|\bar\mu)
\;=\; [F(t)-F(0)] + G(t),
\]
where $G(t):=\sum_{x\in S^c} t\mu_0(x)\ln(t\mu_0(x)/\bar\mu(x))$.
Expanding $G$ using $\bar\mu(x)>0$ on $S^c$,
\[
G(t) \;=\; ct\ln t + O(t), \qquad c:=\sum_{x\in S^c}\mu_0(x)>0.
\]
Since $F(t)-F(0)=O(t)$, we get $D(\mu_t\|\bar\mu)-D(\mu^\star\|\bar\mu)=ct\ln t+O(t)<0$
for all sufficiently small $t>0$, contradicting minimality of $\mu^\star$.
Hence $S^c=\emptyset$ and $\mu^\star>0$ everywhere.

\paragraph{Lagrangian and strong duality.}
Attach a multiplier $\theta$ to the equality constraint $\mathcal{A}\mu=b$ and form
\[
\mathcal{L}(\mu,\theta) \;=\; D(\mu\|\bar\mu) - \langle\theta,\mathcal{A}\mu-b\rangle.
\]
By Proposition~\ref{prop:adjoint}(i)--(ii),
$\langle\theta,\mathcal{A}\mu\rangle=\langle\mathcal{A}^*\theta,\mu\rangle
=\mathbb{E}_\mu[\langle\theta,T\rangle]$, so
\[
\mathcal{L}(\mu,\theta)
\;=\; D(\mu\|\bar\mu) - \mathbb{E}_\mu[\langle\theta,T\rangle] + \langle\theta,b\rangle.
\]
By (A3), $\mu_0\in\mathcal{P}$ satisfies $\mu_0>0$, so
$\mathcal{P}\cap\mathrm{ri}(\Delta(\mathcal{X}))\neq\emptyset$,
where $\mathrm{ri}(\Delta(\mathcal{X}))=\{\mu:\mu(x)>0\ \forall x\}$
denotes the relative interior of the simplex.
This is precisely Slater's condition for the equality-constrained program
with $\Delta(\mathcal{X})$ as the ambient domain: it guarantees strong
duality and that the dual supremum is attained at finite $\theta^\star$
(finite-dimensional convex duality; see \cite{rockafellar1970}, Thm.~28.2).
The minimax theorem (applied to $\mathcal{L}$, which is convex in $\mu$ and
linear in $\theta$) then yields no duality gap:
\begin{equation}
\label{eq:strong_duality}
\min_{\mu\in\mathcal{P}} D(\mu\|\bar\mu)
\;=\; \sup_{\theta}\; g(\theta),
\qquad
g(\theta) \;:=\; \inf_{\mu\in\Delta(\mathcal{X})}\mathcal{L}(\mu,\theta).
\end{equation}

\paragraph{Donsker--Varadhan evaluation of the inner infimum.}
Fix $\theta$ and set $X(x):=\langle\theta,T(x)\rangle$.
Define the tilted measure
\[
\mu_\theta(x)
\;:=\; \frac{\bar\mu(x)\,e^{X(x)}}{\mathbb{E}_{\bar\mu}[e^X]}
\;=\; \bar\mu(x)\,e^{X(x)-\Lambda(\theta)},
\]
using \eqref{eq:logpartition} in the second equality.
For any $\mu\in\Delta(\mathcal{X})$, the identity
$\ln(\mu(x)/\bar\mu(x)) = \ln(\mu(x)/\mu_\theta(x)) + X(x) - \Lambda(\theta)$
gives, after summing against $\mu$,
\[
D(\mu\|\bar\mu) - \mathbb{E}_\mu[X] \;=\; D(\mu\|\mu_\theta) - \Lambda(\theta).
\]
Since $D(\mu\|\mu_\theta)\geq 0$ with equality iff $\mu=\mu_\theta$, the
infimum over $\mu\in\Delta(\mathcal{X})$ equals $-\Lambda(\theta)$ and is
attained uniquely at $\mu_\theta$.  Substituting into \eqref{eq:strong_duality},
\[
g(\theta) \;=\; \langle\theta,b\rangle - \Lambda(\theta),
\]
which is concave in $\theta$ (as $\Lambda$ is convex, being a log-sum-exp).

\paragraph{Gibbs form.}
Since the problem is finite-dimensional, the primal optimum $\mu^\star$ is
attained (proved above) and strong duality holds~\eqref{eq:strong_duality}.
By the KKT theorem for convex programs with linear equality constraints,
there exists a multiplier $\theta^\star\in\mathbb{R}^K$ satisfying the
stationarity condition $\nabla_\mu\mathcal{L}(\mu^\star,\theta^\star)=0$;
this $\theta^\star$ attains the dual supremum.
At any such $\theta^\star$, the inner infimum defining $g(\theta^\star)$ is
uniquely attained at $\mu_{\theta^\star}$ (as shown above); strong duality
then forces $\mu^\star=\mu_{\theta^\star}$, giving
\[
\mu^\star(x) \;\propto\; \bar\mu(x)\,\exp\!\bigl(\langle\theta^\star,T(x)\rangle\bigr).
\]
Expanding $\langle\theta^\star,T(x)\rangle$ block by block via \eqref{eq:MOT_stats}
and Definition~\ref{def:constraint_operators}, with $x=(s_1,v,s_2)$:
\[
\langle\theta^\star,T(x)\rangle
\;=\; u_1(s_1)+u_V(v)+u_2(s_2)
     +\Delta_M(s_1,v)(s_2-s_1)
     +\Delta_C(s_1,v)\bigl(L(s_2/s_1)-v^2\bigr),
\]
where $u_1,u_V,u_2$ are the Lagrange multipliers for the marginal
constraints $\mathcal{A}_{\mathrm{marg}}\mu=\margvec$, and $\Delta_M,\Delta_C$
are the Lagrange multipliers for the financial constraints
$\mathcal{A}_{\mathrm{fin}}\mu=0$.
This is exactly \eqref{eq:MOT_gibbs}.
\end{proof}

\paragraph{Financial interpretation of the dual multipliers.}
\begin{itemize}
  \item $\Delta_M(s_1,v)$ is the \emph{delta hedge multiplier}: it prices the
    conditional forward payoff $s_2-s_1$ at each node $(s_1,v)$, enforcing
    $\mathbb{E}[S_2\mid S_1,V]=S_1$.
  \item $\Delta_C(s_1,v)$ is the \emph{vega hedge multiplier}: it prices the
    conditional log-contract payoff $L(s_2/s_1)-v^2$ at each node $(s_1,v)$,
    enforcing $\mathbb{E}[L(S_2/S_1)\mid S_1,V]=V^2$.
\end{itemize}
Their values are determined by the Newton/LM inner loop of Algorithm~1.

\begin{corollary}[Dual Objective and Optimality Conditions]
\label{cor:dual_optimality}
The optimal value of \eqref{eq:primal} equals
\[
\sup_\theta
\Bigl\{
\langle\margvec,\,\theta_{\mathrm{marg}}\rangle
-
\Lambda(\theta)
\Bigr\},
\]
where $\theta_{\mathrm{marg}} = (u_1,u_V,u_2)$ and the financial components
$(\Delta_M,\Delta_C)$ enter only through $\Lambda(\theta)$.
The supremum is attained at $\theta^\star$ whose first-order conditions are
\begin{equation}
\label{eq:FOC_all}
\frac{\partial\Lambda}{\partial\theta_k}(\theta^\star)
= b_k
\quad\forall\,k,
\end{equation}
where $b_k = \margvec_k$ for marginal components and $b_k = 0$ for financial
components.  Equation~\eqref{eq:FOC_all} is precisely $\mathbb{E}_{\mu^\star}[T_k] = b_k$,
i.e., the calibrated coupling matches all constraints.
\end{corollary}
\begin{proof}
We derive the dual formula, establish the first-order conditions at the
optimum, and interpret them block by block following
\citeauthor{guyon2020spxvix}~\cite{guyon2020spxvix}, eq.~(5.8).

\paragraph{Dual objective.}
From the proof of Theorem~\ref{thm:existence}, strong duality holds and
the dual function is
\[
g(\theta)
\;=\;\inf_{\mu\in\Delta(\mathcal{X})}\mathcal{L}(\mu,\theta)
\;=\;\langle\theta,b\rangle - \Lambda(\theta).
\]
Since $b=(\margvec;\,0)$, the inner product splits as
$\langle\theta,b\rangle
=\langle\theta_{\mathrm{marg}},\margvec\rangle
+\langle\theta_{\mathrm{fin}},0\rangle
=\langle\margvec,\theta_{\mathrm{marg}}\rangle$,
so the dual problem is
\begin{equation}
\label{eq:dual_explicit}
D(\mu^\star\|\bar\mu)
\;=\;\sup_\theta\,\bigl\{\langle\margvec,\theta_{\mathrm{marg}}\rangle - \Lambda(\theta)\bigr\},
\end{equation}
with $g$ concave (as $\Lambda$ is convex).

\paragraph{First-order conditions.}
Since the dual supremum is attained at $\theta^\star$
(Theorem~\ref{thm:existence}) and $g$ is differentiable,
$\nabla_\theta g(\theta^\star)=0$, i.e.\
\begin{equation}
\label{eq:FOC_derive}
b_k \;=\; \frac{\partial\Lambda}{\partial\theta_k}(\theta^\star)
\qquad\forall\,k.
\end{equation}

\paragraph{Mean-parameterization identity.}
Differentiating $\Lambda(\theta)=\ln\sum_x\bar\mu(x)e^{\langle\theta,T(x)\rangle}$
with respect to $\theta_k$ gives
\[
\frac{\partial\Lambda}{\partial\theta_k}(\theta)
=\frac{\sum_x T_k(x)\,\bar\mu(x)\,e^{\langle\theta,T(x)\rangle}}
      {\sum_x \bar\mu(x)\,e^{\langle\theta,T(x)\rangle}}
=\mathbb{E}_{\mu_\theta}[T_k],
\]
where $\mu_\theta(x)\propto\bar\mu(x)e^{\langle\theta,T(x)\rangle}$.
At the dual optimum $\mu_{\theta^\star}=\mu^\star$
(Theorem~\ref{thm:existence}), so \eqref{eq:FOC_derive} becomes
\[
\mathbb{E}_{\mu^\star}[T_k] \;=\; b_k \qquad\forall\,k,
\]
which is exactly~\eqref{eq:FOC_all}.

\paragraph{Block-wise interpretation.}
Reading \eqref{eq:FOC_all} block by block via \eqref{eq:MOT_stats}:
\begin{itemize}
\item \emph{Marginal block} ($k=(1,i),(V,j),(2,k')$):
      $\mathbb{E}_{\mu^\star}[\mathbf{1}_{s_1=s_1^i}]=\mu_1(s_1^i)$,\;
      $\mathbb{E}_{\mu^\star}[\mathbf{1}_{v=v^j}]=\mu_V(v^j)$,\;
      $\mathbb{E}_{\mu^\star}[\mathbf{1}_{s_2=s_2^{k'}}]=\mu_2(s_2^{k'})$:
      the calibrated coupling reproduces the three prescribed marginals.
\item \emph{Delta block} ($k=(M,i,j)$):
      $\mathbb{E}_{\mu^\star}\!\bigl[\mathbf{1}_{(s_1,v)=(s_1^i,v^j)}(s_2-s_1^i)\bigr]=0$:
      the martingale condition holds at every node $(s_1^i,v^j)$.
\item \emph{Vega block} ($k=(C,i,j)$):
      $\mathbb{E}_{\mu^\star}\!\bigl[\mathbf{1}_{(s_1,v)=(s_1^i,v^j)}
      (L(s_2/s_1^i)-(v^j)^2)\bigr]=0$:
      the variance-consistency condition holds at every node $(s_1^i,v^j)$.
\end{itemize}
Together these confirm that $\mu^\star\in\mathcal{P}$, so the calibrated
coupling matches all constraints simultaneously.
\end{proof}

\subsection{Gauge Fixing And Numerical Stability}

The marginal potentials $(u_1,u_V,u_2)$ are not uniquely determined.
Adding constants $(c_1,c_V,c_2)$ with $c_1+c_V+c_2=0$ leaves the
unnormalized Gibbs weights $\bar\mu(x)e^{\langle\theta,T(x)\rangle}$
unchanged; equivalently, since the normalization factor $Z(\theta)$ scales
by $e^{c_1+c_V+c_2}$, the \emph{normalized} coupling~\eqref{eq:MOT_gibbs}
is invariant under \emph{any} constant shift $(c_1,c_V,c_2)$.
The two-dimensional gauge freedom
is confined to the marginal block; the financial multipliers $\Delta_M(s_1,v)$
and $\Delta_C(s_1,v)$ are \emph{uniquely determined} by the financial constraints
at each node $(s_1^i,v^j)$ and carry no gauge ambiguity.

To fix the gauge in the marginal block we impose
\[
\sum_{s_1}\mu_1(s_1)\,u_1(s_1)=0,
\qquad
\sum_{v}\mu_V(v)\,u_V(v)=0.
\]

Under this fixing, the full gauge-fixed parameter vector is
\[
\theta\;\in\;\mathbb{R}^{(N_1+N_V+N_2-2)\,+\,2N_1N_V},
\]
where $N_1+N_V+N_2-2$ degrees of freedom come from the gauge-fixed
marginal potentials and $2N_1N_V$ from the financial multipliers
$\{\Delta_M^{ij},\Delta_C^{ij}\}$.

\subsection{Fisher Information: Hessian of the Log-Partition and Covariance of Dual Basis Functions}
\label{subsec:fisher}

We now prove the central identity that underlies the perturbation theory,
and state the non-degeneracy condition under which $H$ is strictly positive
definite.

\begin{assumption}[Conditional non-degeneracy]\label{ass:non_deg}
For every node $(s_1^i,v^j)\in\mathcal{S}_1\times\mathcal{V}$, the
conditional distribution $\mu^\star_{ij}(\cdot):=\mu^\star(\cdot\mid
s_1^i,v^j)$ has support of cardinality $\geq 3$, and the three functions
$\{1,\;s_2-s_1^i,\;L(s_2/s_1^i)-(v^j)^2\}$ are linearly independent on
$\mathrm{supp}(\mu^\star_{ij})$.
\end{assumption}

\begin{remark}
Assumption~\ref{ass:non_deg} is a condition on the grid, not on the
coupling: it holds whenever $|\mathcal{S}_2|\geq 3$ and no three
points in $\mathcal{S}_2$ satisfy an exact affine relation involving the
delta and vega payoff kernels.
Since $\bar\mu>0$ forces $\mu^\star>0$ on $\mathcal{X}$, the full support
condition $|\mathrm{supp}(\mu^\star_{ij})|=|\mathcal{S}_2|$ is automatic;
the affine-independence part is a mild regularity condition satisfied by
any generic grid.
In practice (and in the numerical experiments of Section~\ref{sec:experiments}),
Assumption~\ref{ass:non_deg} holds, and $H$ is strictly positive definite.
When it fails at isolated nodes, $H$ is positive semi-definite and the
Moore--Penrose pseudoinverse $H^+$ replaces $H^{-1}$ throughout;
all downstream results remain valid (see Lemmas~\ref{lem:range_gG}
and~\ref{lem:range_db} below).
\end{remark}

\paragraph{Log-partition function and invertibility of $H$.}
On the gauge-fixed parameter space, $\Lambda(\theta)$ is $C^\infty$.
Its Hessian $H = \nabla^2_\theta\Lambda(\theta^\star)$ coincides with
the Fisher information matrix of the calibrated exponential family.
$H$ is positive semi-definite in general; under Assumption~\ref{ass:non_deg}
above, $H$ is strictly positive definite and the linear response system
has a unique solution.
The Moore--Penrose pseudoinverse $H^+$ is used throughout;
under Assumption~\ref{ass:non_deg}, $H^+=H^{-1}$.

\begin{proposition}[Fisher Information = Hessian of $\Lambda$ = Covariance of $T$]
\label{prop:fisher_hessian_cov}
Let $\theta$ be the gauge-fixed MOT parameter vector, $\{T_k\}$ the
sufficient statistics~\eqref{eq:MOT_stats}, and
$\mu_\theta(x)\propto\bar\mu(x)\exp\langle\theta,T(x)\rangle$.  Then:
\begin{enumerate}
  \item[(i)] \emph{(Mean parameterization)}
    $\displaystyle\frac{\partial\Lambda}{\partial\theta_k}
       = \mathbb{E}_{\mu_\theta}[T_k].$
  \item[(ii)] \emph{(Fisher information as Hessian and covariance)}
    $\displaystyle
       H_{kl} := \frac{\partial^2\Lambda}{\partial\theta_k\,\partial\theta_l}
       = \mathrm{Cov}_{\mu_\theta}(T_k,\,T_l).$
\end{enumerate}
In particular $H=\nabla^2_\theta\Lambda\succeq 0$.
Under Assumption~\ref{ass:non_deg}, $H\succ 0$ on the gauge-fixed
subspace.
\end{proposition}

\begin{proof}
Write the partition function and log-partition function explicitly:
\[
Z(\theta)
\;=\;\sum_{x\in\mathcal{X}}\bar\mu(x)\,e^{\langle\theta,\,T(x)\rangle},
\qquad
\Lambda(\theta) = \ln Z(\theta).
\]
For $x=(s_1,v,s_2)$, expanding $\langle\theta,T(x)\rangle$ block by block
via~\eqref{eq:MOT_stats}:
\[
\langle\theta,T(x)\rangle
\;=\;
u_1(s_1)+u_V(v)+u_2(s_2)
\;+\;\Delta_M(s_1,v)\,(s_2-s_1)
\;+\;\Delta_C(s_1,v)\,\bigl(L(s_2/s_1)-v^2\bigr),
\]
so $Z(\theta)=\sum_{s_1,v,s_2}\bar\mu(s_1,v,s_2)\,e^{\langle\theta,T(x)\rangle}$
and the Gibbs measure is
$\mu_\theta(x)=\bar\mu(x)\,e^{\langle\theta,T(x)\rangle}/Z(\theta)$.

\paragraph{Part (i): gradient of $\Lambda$.}
Differentiating $\Lambda=\ln Z$ with respect to $\theta_k$ by the chain rule:
\[
\frac{\partial\Lambda}{\partial\theta_k}(\theta)
\;=\;\frac{1}{Z(\theta)}\,\frac{\partial Z(\theta)}{\partial\theta_k}.
\]
Since $T_k(x)$ multiplies $e^{\langle\theta,T(x)\rangle}$ when differentiating
with respect to $\theta_k$,
\[
\frac{\partial Z(\theta)}{\partial\theta_k}
\;=\;\sum_{x\in\mathcal{X}}\bar\mu(x)\,T_k(x)\,e^{\langle\theta,T(x)\rangle}.
\]
Dividing through by $Z(\theta)$ recognizes each term as $\mu_\theta(x)$:
\[
\frac{\partial\Lambda}{\partial\theta_k}(\theta)
\;=\;\sum_{x\in\mathcal{X}}T_k(x)\,
      \underbrace{\frac{\bar\mu(x)\,e^{\langle\theta,T(x)\rangle}}{Z(\theta)}}_{\mu_\theta(x)}
\;=\;\mathbb{E}_{\mu_\theta}[T_k].
\]

\paragraph{Part (ii): Hessian of $\Lambda$.}
Differentiate $\partial\Lambda/\partial\theta_k = \mathbb{E}_{\mu_\theta}[T_k]
= \sum_x T_k(x)\mu_\theta(x)$ with respect to $\theta_l$:
\[
\frac{\partial^2\Lambda}{\partial\theta_k\,\partial\theta_l}
\;=\;\sum_{x\in\mathcal{X}} T_k(x)\,\frac{\partial\mu_\theta(x)}{\partial\theta_l}.
\]
From $\mu_\theta(x)=\bar\mu(x)e^{\langle\theta,T(x)\rangle}/Z(\theta)$,
\[
\frac{\partial\mu_\theta(x)}{\partial\theta_l}
\;=\;\mu_\theta(x)\!\left(T_l(x)-\frac{\partial\Lambda}{\partial\theta_l}\right)
\;=\;\mu_\theta(x)\bigl(T_l(x)-\mathbb{E}_{\mu_\theta}[T_l]\bigr).
\]
Substituting and splitting the sum:
\begin{align*}
\frac{\partial^2\Lambda}{\partial\theta_k\,\partial\theta_l}
&\;=\;\sum_x T_k(x)\,\mu_\theta(x)\bigl(T_l(x)-\mathbb{E}_{\mu_\theta}[T_l]\bigr)\\
&\;=\;\sum_x T_k(x)\,T_l(x)\,\mu_\theta(x)
     \;-\;\mathbb{E}_{\mu_\theta}[T_l]\sum_x T_k(x)\,\mu_\theta(x)\\
&\;=\;\mathbb{E}_{\mu_\theta}[T_kT_l]
     \;-\;\mathbb{E}_{\mu_\theta}[T_k]\,\mathbb{E}_{\mu_\theta}[T_l]
\;=\;\mathrm{Cov}_{\mu_\theta}(T_k,T_l).
\end{align*}

\paragraph{Positive semi-definiteness.}
For any $c\in\mathbb{R}^K$,
\[
c^\top H c
\;=\;\mathrm{Var}_{\mu_\theta}\!\bigl(\langle c,T\rangle\bigr)\;\geq\;0,
\]
so $H\succeq 0$.  Equality holds iff $\sum_k c_k T_k(x)$ is constant on
$\mathcal{X}$, i.e.\ $c\in\ker(H)$.  Under Assumption~\ref{ass:non_deg}
the three functions $\{1,s_2-s_1^i,L(s_2/s_1^i)-(v^j)^2\}$ are linearly
independent on each fiber, so no non-trivial combination in $\mathbb{R}^K$
is constant; together with gauge fixing, this forces $c=0$, giving
$H\succ 0$.
\end{proof}

\begin{lemma}[Payoff covariances lie in $\mathrm{range}(H)$]
\label{lem:range_gG}
For any payoff $G:\mathcal{X}\to\mathbb{R}$, the covariance vector
$g_G := \mathrm{Cov}_{\mu^\star}(G,T)\in\mathbb{R}^K$ satisfies
$g_G\in\mathrm{range}(H)$.
\end{lemma}
\begin{proof}
If $c\in\ker(H)$ then $\sum_k c_k T_k\equiv\gamma$ on $\mathcal{X}$,
so $c^\top g_G = \mathrm{Cov}_{\mu^\star}(G,\gamma)=0$.
Hence $g_G\perp\ker(H) = \mathrm{range}(H)^\perp$.
\end{proof}

\begin{lemma}[Admissible perturbations lie in $\mathrm{range}(H)$]
\label{lem:range_db}
Let $\dot b=(h_1,h_V,h_2,0,\ldots,0)\in\mathbb{R}^K$ be a marginal
perturbation satisfying the mass-conservation conditions
$\sum_{s_1}h_1=\sum_v h_V=\sum_{s_2}h_2=0$.
Then $\dot b\in\mathrm{range}(H)$, so $H\dot\theta^0=\dot b$ is solvable.
\end{lemma}
\begin{proof}
Any $c\in\ker(H)$ has its marginal block of the form
$({\gamma_1}\mathbf{1},{\gamma_V}\mathbf{1},{\gamma_2}\mathbf{1})$ with
$\gamma_1+\gamma_V+\gamma_2=0$ (from the gauge structure).
Then $c^\top\dot b = \gamma_1\sum h_1+\gamma_V\sum h_V+\gamma_2\sum h_2=0$
by mass conservation, so $\dot b\perp\ker(H)$.
\end{proof}

\begin{remark}
This identity encodes three equivalent viewpoints on $H$:
(a)~\emph{curvature of the dual} $H$ measures how sharply the dual objective
    curves at the optimum;
(b)~\emph{spread of the dual basis functions} $H_{kl}$ is the covariance of
    the sufficient statistics $T_k$ and $T_l$ under the calibrated Gibbs measure;
(c)~\emph{Riemannian metric}-$H$ is the Fisher--Rao metric on the statistical
    manifold of calibrated couplings.
The linear response formula $\Pi'(0)=g^\top H^+h$
(Theorem~\ref{thm:risk_representation}) is the statement that marginal
perturbations $h$ project onto the tangent space of this manifold via $H^+$.
\end{remark}

\paragraph{Block structure of the MOT Fisher matrix.}
Partitioning $H$ according to $\theta=(u,\Delta_M,\Delta_C)$ gives
\begin{equation}
\label{eq:MOT_fisher_block}
H =
\begin{pmatrix}
  H_{\mathrm{marg}} & H_{\mathrm{marg},M} & H_{\mathrm{marg},C} \\[4pt]
  H_{M,\mathrm{marg}} & H_{MM} & H_{MC} \\[4pt]
  H_{C,\mathrm{marg}} & H_{CM} & H_{CC}
\end{pmatrix},
\end{equation}
where $\mu^\star_{ij}$ denotes the conditional distribution of $s_2$ given
$(s_1,v)=(s_1^i,v^j)$, and the financial diagonal blocks are block-diagonal:
\begin{align*}
  (H_{MM})_{ij,kl}
    &= \mathbf{1}_{(i,j)=(k,l)}\;\mathrm{Var}_{\mu^\star_{ij}}(s_2-s_1^i),\\
  (H_{CC})_{ij,kl}
    &= \mathbf{1}_{(i,j)=(k,l)}\;
       \mathrm{Var}_{\mu^\star_{ij}}\!\bigl(L(s_2/s_1^i)-(v^j)^2\bigr),\\
  (H_{MC})_{ij,kl}
    &= \mathbf{1}_{(i,j)=(k,l)}\;
       \mathrm{Cov}_{\mu^\star_{ij}}\!\bigl(s_2-s_1^i,\;
       L(s_2/s_1^i)-(v^j)^2\bigr).
\end{align*}
The block-diagonal structure of $H_{MM}$ and $H_{CC}$ reflects the
node-local support of the financial payoffs; only the off-diagonal blocks
$H_{\mathrm{marg},M}$ and $H_{\mathrm{marg},C}$ couple the marginal and
financial directions.

\section{Perturbation Theory: General Marginal Shocks}
\label{sec:perturbation}

\subsection{Admissible Perturbations}

Let $\mu^\star$ be the unique MOT optimizer of Section~\ref{sec:framework}.
Risk generation requires understanding how $\mu^\star$ responds to changes in
market data.  In the SPX--VIX setting, market perturbations affect the
\emph{marginal} constraints (changing option prices shifts $\mu_1,\mu_V,\mu_2$)
while the \emph{financial} constraints (martingality and variance consistency)
are structural and remain fixed.

A \emph{directional marginal perturbation} specifies component functions
\[
h_1:\mathcal{S}_1\to\mathbb{R},\quad h_V:\mathcal{V}\to\mathbb{R},\quad h_2:\mathcal{S}_2\to\mathbb{R},
\]
satisfying the mass-preserving constraints
\[
\textstyle\sum_{s_1} h_1(s_1)=0,\quad
\sum_{v} h_V(v)=0,\quad
\sum_{s_2} h_2(s_2)=0.
\]
For $\varepsilon$ sufficiently small, the perturbed marginals
$\mu_i^\varepsilon = \mu_i + \varepsilon h_i$ remain in the interior of the simplex.
In the full parameter space $\theta = (u_1,u_V,u_2,\{\Delta_M^{ij}\},\{\Delta_C^{ij}\})$
this corresponds to the full perturbation vector
\begin{equation}
\label{eq:perturbation_vector}
h \;:=\; \bigl(h_1,\,h_V,\,h_2,\;\underbrace{0,\ldots,0}_{2N_1N_V}\bigr)\;\in\mathbb{R}^K,
\end{equation}
i.e., the financial constraint slots remain zero.
We set $\dot b := h$; the two notations are equivalent and we use $h$ henceforth.

Here $h_{\mathrm{marg}}:=(h_1,h_V,h_2)$ denotes the marginal block of $h$.
The perturbed MOT-admissible set is
\begin{equation}
\label{eq:primal_eps}
\mathcal{P}^\varepsilon
= \bigl\{\mu\in\Delta(\mathcal{X}):\mathcal{A}_{\mathrm{marg}}\mu=\margvec+\varepsilon h_{\mathrm{marg}},\;
  \mathcal{A}_{\mathrm{fin}}\mu=0\bigr\},
\end{equation}
and we denote the unique optimizer $\mu^\varepsilon=\arg\inf_{\mu\in\mathcal{P}^\varepsilon}D(\mu\|\bar\mu)$.

Define the \emph{tangent space of financially-consistent mass-preserving variations}
\begin{equation}
\label{eq:tangent_V}
V \;:=\; \ker(\mathcal{A}_{\mathrm{fin}}) \cap \ker(\mathbf{1}^\top)
  \;=\; \bigl\{v\in\mathbb{R}^{\mathcal{X}} :\mathcal{A}_{\mathrm{fin}}v=0,\;
        \textstyle\sum_x v(x)=0\bigr\},
\end{equation}
and write $W_0:=\{(a_1,a_V,a_2)\in\mathbb{R}^{N_1+N_V+N_2}:\sum a_1=\sum a_V=\sum a_2=0\}$
for the subspace of mass-preserving marginal targets.

\begin{lemma}[Surjectivity of the reduced constraint map]
\label{lem:surjectivity}
Under Assumption~\ref{ass:non_deg}, the map
$\mathcal{A}_{\mathrm{marg}}|_V : V \to W_0$ is surjective.
\end{lemma}
\begin{proof}
By the finite-dimensional rank theorem it suffices to show the transpose is
injective modulo the natural annihilator.  Suppose
$(y_1,y_V,y_2)\in W_0^\perp$ satisfies
$(\mathcal{A}_{\mathrm{marg}}|_V)^\top(y_1,y_V,y_2) = 0$, i.e.\
$\mathcal{A}_{\mathrm{marg}}^\top(y_1,y_V,y_2)\in V^\perp =
\mathrm{range}(\mathcal{A}_{\mathrm{fin}}^\top)+\mathbb{R}\mathbf{1}$.
Expanding, there exist $\alpha_{ij},\beta_{ij}\in\mathbb{R}$ and $c\in\mathbb{R}$
such that for all $(i,j,k)$,
\begin{equation}
\label{eq:surj_annihilator}
y_1(s_1^i)+y_V(v^j)+y_2(s_2^k)
\;=\; c + \alpha_{ij}(s_2^k-s_1^i)
     + \beta_{ij}\bigl(L(s_2^k/s_1^i)-(v^j)^2\bigr).
\end{equation}
Fix $(i,j)$ and vary $k$.  The left-hand side varies only through $y_2(s_2^k)$,
so the right-hand side is $\alpha_{ij}s_2^k + \beta_{ij}L(s_2^k/s_1^i)+\mathrm{const}$
as a function of $s_2^k$.
By Assumption~\ref{ass:non_deg} the functions $\{1,\,s_2-s_1^i,\,L(s_2/s_1^i)-(v^j)^2\}$
are linearly independent on $\mathcal{S}_2$, so we conclude $\alpha_{ij}=\beta_{ij}=0$
and $y_2$ is constant.
Substituting back into~\eqref{eq:surj_annihilator} gives $y_1(s_1^i)+y_V(v^j)=\mathrm{const}$
for all $(i,j)$, which forces $y_1$ and $y_V$ to be constant as well.
The block-sum conditions defining $W_0^\perp$ then force each constant to be zero,
so $(y_1,y_V,y_2)=0$.
\end{proof}

\begin{theorem}[Well-posedness under marginal perturbations]
\label{thm:wellposed}
Suppose $\mu^\star>0$ and Assumption~\ref{ass:non_deg} holds.
Let $(h_1,h_V,h_2)$ be any mass-preserving marginal perturbation~\eqref{eq:perturbation_vector}.
Then there exists $\varepsilon_0>0$ such that for all $|\varepsilon|<\varepsilon_0$
the perturbed feasible set $\mathcal{P}^\varepsilon$~\eqref{eq:primal_eps} is non-empty
and the perturbed problem has a unique minimizer $\mu^\varepsilon$.
\end{theorem}
\begin{proof}
Since $h_{\mathrm{marg}}=(h_1,h_V,h_2)\in W_0$ (by mass preservation), Lemma~\ref{lem:surjectivity}
yields $\dot\mu\in V$ with $\mathcal{A}_{\mathrm{marg}}\dot\mu=h_{\mathrm{marg}}$.
Set $\mu(\varepsilon):=\mu^\star+\varepsilon\dot\mu$.
By construction, $\mathcal{A}_{\mathrm{marg}}\mu(\varepsilon)=\margvec+\varepsilon h_{\mathrm{marg}}$,
$\mathcal{A}_{\mathrm{fin}}\mu(\varepsilon)=0$ (since $\dot\mu\in\ker\mathcal{A}_{\mathrm{fin}}$),
and $\mathbf{1}^\top\mu(\varepsilon)=1$.
Since $\mu^\star>0$, there exists $\varepsilon_0>0$ such that $\mu(\varepsilon)>0$
for all $|\varepsilon|<\varepsilon_0$, giving $\mu(\varepsilon)\in\mathcal{P}^\varepsilon$.
Uniqueness of the minimizer follows from strict convexity of $D(\cdot\|\bar\mu)$
on the convex set $\mathcal{P}^\varepsilon$.
\end{proof}

\subsection{Perturbed Dual Variables And MOT Gibbs Form}

By Corollary~\ref{cor:dual_optimality}, the optimizer $\mu^\varepsilon$ takes the full
MOT Gibbs form.  There exist perturbed dual parameters
$\theta^\varepsilon = \bigl(u_1^\varepsilon,u_V^\varepsilon,u_2^\varepsilon,
\{\Delta_M^\varepsilon(s_1,v)\},\{\Delta_C^\varepsilon(s_1,v)\}\bigr)$
such that
\begin{equation}
\label{eq:mu_eps_scaling}
\mu^\varepsilon(s_1,v,s_2)
\;\propto\;
\bar\mu(s_1,v,s_2)\,
\exp\!\Bigl(
  u_1^\varepsilon(s_1)+u_V^\varepsilon(v)+u_2^\varepsilon(s_2)
  +\Delta_M^\varepsilon(s_1,v)(s_2-s_1)
  +\Delta_C^\varepsilon(s_1,v)\bigl(L(s_2/s_1)-v^2\bigr)
\Bigr).
\end{equation}
At $\varepsilon=0$ this reduces to the base coupling $\mu^\star$ of
\eqref{eq:MOT_gibbs}.  It is convenient to collect the dual parameters as
\[
\theta^\varepsilon = \theta^\star + \varepsilon\,\dot\theta^0 + o(\varepsilon),
\]
where $\theta^\star = (u^\star_1,u^\star_V,u^\star_2,\{\Delta_M^{\star,ij}\},\{\Delta_C^{\star,ij}\})$ are the base calibrated dual parameters.
Working in \emph{scaling variables}
$a^\varepsilon(s_1)=e^{u_1^\varepsilon(s_1)}$,
$q^\varepsilon(v)=e^{u_V^\varepsilon(v)}$,
$c^\varepsilon(s_2)=e^{u_2^\varepsilon(s_2)}$
the Gibbs form becomes
\begin{equation}
\label{eq:scaling_form}
\mu^\varepsilon(s_1,v,s_2)\propto
\bar\mu(s_1,v,s_2)\,a^\varepsilon(s_1)q^\varepsilon(v)c^\varepsilon(s_2)\,
e^{\Delta_M^\varepsilon(s_1,v)(s_2-s_1)+\Delta_C^\varepsilon(s_1,v)(L(s_2/s_1)-v^2)}.
\end{equation}

\paragraph{Gauge fixing.}
Only the marginal potentials $(u_1^\varepsilon,u_V^\varepsilon,u_2^\varepsilon)$
carry gauge freedom; the financial multipliers are uniquely determined.  We fix
\begin{equation}
\label{eq:gauge}
\sum_{s_1}\mu_1(s_1)\,u_1^\varepsilon(s_1)=0,\qquad
\sum_{v}\mu_V(v)\,u_V^\varepsilon(v)=0,
\end{equation}
pinning down the full parameter vector $\theta^\varepsilon$ uniquely.

\subsection{Differentiability Of The Entropic MOT Projection}

We prove that the full MOT optimizer $(\mu^\varepsilon,\theta^\varepsilon)$ depends
smoothly on $\varepsilon$ and derive the first-order linear system.

The log-partition function $\Lambda(\theta)$ of~\eqref{eq:logpartition} is
$C^\infty$ and strictly convex on the gauge-fixed parameter space (by
Proposition~\ref{prop:fisher_hessian_cov}).  The full dual objective is
\[
\mathcal{D}(\theta;\varepsilon)
:= \langle\margvec^\varepsilon,\,\theta_{\mathrm{marg}}\rangle - \Lambda(\theta),
\]
where $\margvec^\varepsilon = (\mu_1^\varepsilon,\mu_V^\varepsilon,\mu_2^\varepsilon)$ and
the financial components $(\Delta_M,\Delta_C)$ enter only through $\Lambda(\theta)$.

\begin{lemma}[Structural constant-rank of $H$]
\label{lem:const_rank}
Define the null subspace
$N:=\{v\in\mathbb{R}^K : \langle v,T(x)\rangle=\mathrm{const}\ \bar\mu\text{-a.e.}\}$.
Then $\ker H(\theta)=N$ for every $\theta\in\mathbb{R}^K$, so
$\mathrm{rank}\,H(\theta)$ is constant.
\end{lemma}
\begin{proof}
$H(\theta)=\mathrm{Cov}_{Q_\theta}(T)$ with $Q_\theta\propto\bar\mu\,e^{\langle\theta,T\rangle}$.
Since $\bar\mu>0$, $Q_\theta\sim\bar\mu$ for every $\theta$.
Hence $\mathrm{Var}_{Q_\theta}(\langle v,T\rangle)=0$ iff $\langle v,T(x)\rangle$
is constant $Q_\theta$-a.s., iff it is constant $\bar\mu$-a.s., i.e.\ $v\in N$.
\end{proof}

\begin{theorem}[Differentiability of MOT optimal parameters and coupling]
\label{thm:differentiability}
Fix the gauge~\eqref{eq:gauge} and assume $\mu^\star>0$.
Then there exists $\varepsilon_0>0$ such that on $(-\varepsilon_0,\varepsilon_0)$:
\begin{enumerate}
\item $\varepsilon\mapsto\theta^\varepsilon$ is $C^\infty$,
\item $\varepsilon\mapsto\mu^\varepsilon$ is $C^\infty$ entrywise,
\item the first-order variation $\dot\theta^0 := d\theta^\varepsilon/d\varepsilon|_{\varepsilon=0}$
      is the unique element of $N^\perp=\mathrm{range}(H)$ solving
      \begin{equation}
      \label{eq:linear_system_full}
      H\,\dot\theta^0 = h,
      \end{equation}
      where $H=\nabla^2_\theta\Lambda(\theta^\star)$ is the full MOT Fisher matrix
      \eqref{eq:MOT_fisher_block},
      $h = (h_1,h_V,h_2,0_{N_1N_V},0_{N_1N_V})$
      is the perturbation vector~\eqref{eq:perturbation_vector},
      and the minimum-norm solution is $\dot\theta^0=H^+h$.
\end{enumerate}
\end{theorem}

\begin{proof}
The full KKT optimality conditions for the perturbed dual read
\begin{equation}
\label{eq:FOC_dual}
\nabla_\theta\Lambda(\theta^\varepsilon) = b^\varepsilon,
\end{equation}
where $b^\varepsilon = (\margvec^\varepsilon, 0_{N_1N_V}, 0_{N_1N_V})$ stacks the
perturbed marginal targets and the (unchanged) financial targets.

By Lemma~\ref{lem:const_rank}, $\ker H(\theta)=N$ is independent of $\theta$,
so $\Lambda$ is invariant under $\theta\mapsto\theta+v$ for any $v\in N$ and
descends to a $C^\infty$ strictly convex function $\bar\Lambda$ on $N^\perp\cong\mathbb{R}^K/N$.
The compression of $H$ to $N^\perp$ is invertible by
Proposition~\ref{prop:fisher_hessian_cov}.
By Lemma~\ref{lem:range_db}, $b^\varepsilon - b^0\in N^\perp$ for small $|\varepsilon|$,
so $\nabla\bar\Lambda([\theta^\varepsilon])=[b^\varepsilon]$ is a well-posed equation on $N^\perp$.
The classical implicit function theorem applied to this equation on $N^\perp$
yields a unique $C^\infty$ curve $\varepsilon\mapsto\theta^\varepsilon\in N^\perp$.
Under Assumption~\ref{ass:non_deg}, $N=\{0\}$ and this reduces to the
standard IFT statement.

Differentiating~\eqref{eq:FOC_dual} at $\varepsilon=0$ gives
$H\,\dot\theta^0 = h$,
where $h = db^\varepsilon/d\varepsilon|_{\varepsilon=0}$ is the
perturbation vector~\eqref{eq:perturbation_vector}.
The minimum-norm solution $\dot\theta^0 = H^+h$ gives the first-order change
in \emph{all} dual parameters: marginal potentials \emph{and} financial
multipliers simultaneously adjust to maintain the constraints.
Differentiability of $\mu^\varepsilon$ follows from the smooth Gibbs map~\eqref{eq:mu_eps_scaling}.
\end{proof}

\begin{remark}
The financial multiplier components of $\dot\theta^0$, namely
$\dot\Delta_M^0$ and $\dot\Delta_C^0$, are determined by the off-diagonal
blocks $H_{\mathrm{marg},M}$ and $H_{\mathrm{marg},C}$ of the Fisher matrix.
They quantify how the delta and vega hedges must re-adjust to preserve
martingality and variance consistency as the SPX/VIX marginals change.
\end{remark}

\begin{remark}[Lift-independence of the risk formula]
\label{rem:lift_independence}
The curve $\theta^\varepsilon\in N^\perp$ constructed above is a canonical
representative: any other solution of $\nabla_\theta\Lambda(\theta^\varepsilon)=b^\varepsilon$
differs by an element of $N$.
Since $g_G=\mathrm{Cov}_{\mu^\star}(G,T)\in\mathrm{range}(H)=N^\perp$ by
Lemma~\ref{lem:range_gG}, the inner product $g_G^\top\theta^\varepsilon$ is
independent of the choice of representative and the risk formula
$\Pi'(0)=g_G^\top H^+h$ is unambiguous.
\end{remark}

\subsection{Risk Representation: Gâteaux Derivative Of Expectations}
\label{subsec:risk_rep}

Let $G:\mathcal{X}\to\mathbb{R}$ be any payoff (bounded is automatic since $\mathcal{X}$ finite).
Define the model price under calibration $\mu^\varepsilon$ by
\[
\Pi(\varepsilon) := \mathbb{E}_{\mu^\varepsilon}[G] = \sum_{x\in\mathcal{X}} G(x)\mu^\varepsilon(x).
\]

\begin{theorem}[General risk representation]
\label{thm:risk_representation}
Let $G:\mathcal{X}\to\mathbb{R}$ be any payoff and let
$\Pi(\varepsilon)=\mathbb{E}_{\mu^{\varepsilon}}[G]$.
Under the assumptions of Theorem~\ref{thm:differentiability},
$\varepsilon\mapsto\Pi(\varepsilon)$ is $C^1$ and
\begin{equation}
\label{eq:risk_representation}
\Pi'(0)
=\langle h_{1},\psi_{1}\rangle
+\langle h_{V},\psi_{V}\rangle
+\langle h_{2},\psi_{2}\rangle,
\end{equation}
where $\psi=(\psi_{1},\psi_{V},\psi_{2})$ are the \emph{marginal influence
functions} of $G$.  The full influence vector
\begin{equation}
\label{eq:influence_full}
\Psi := H^+\, g_{G}
\end{equation}
has components indexed by all sufficient statistics~\eqref{eq:MOT_stats},
where $H=\nabla^2_\theta\Lambda(\theta^\star)$ is the full MOT Fisher
matrix~\eqref{eq:MOT_fisher_block} and
\[
(g_{G})_{k} = \operatorname{Cov}_{\mu^{\star}}(G,\,T_{k})
\]
runs over $\{T_k\}$ = marginal indicators $\cup$ delta payoffs $\cup$ vega payoffs.
The marginal influence functions $\psi$ are the components of $\Psi$
corresponding to the marginal sufficient statistics.
\end{theorem}

\begin{proof}
Since $\mu^{\varepsilon}$ is $C^{1}$ in $\varepsilon$ by Theorem~\ref{thm:differentiability} and
$\mathcal{X}$ is finite, the price
\[
\Pi(\varepsilon)=\sum_{x\in\mathcal{X}} G(x)\,\mu^{\varepsilon}(x)
\]
is $C^{1}$ and
\[
\Pi'(0)=\sum_{x} G(x)\,\dot{\mu}^{0}(x),
\]
where $\dot{\mu}^{0}$ denotes $\frac{d}{d\varepsilon}\mu^{\varepsilon}\big|_{\varepsilon=0}$.

From the full MOT Gibbs representation \eqref{eq:mu_eps_scaling},
\[
\mu^{\varepsilon}(x)
= \exp\!\bigl(\log\bar{\mu}(x)
  + \langle\theta^\varepsilon,T(x)\rangle
  - \Lambda(\theta^{\varepsilon})\bigr),
\]
the standard exponential-family differentiation gives
\[
\frac{\dot{\mu}^{0}(x)}{\mu^{\star}(x)}
= \langle\dot\theta^0,T(x)\rangle
 - \mathbb{E}_{\mu^{\star}}[\langle\dot\theta^0,T\rangle],
\]
where $T(x)$ is the full sufficient statistic vector~\eqref{eq:MOT_stats}
and $\dot\theta^0 = d\theta^\varepsilon/d\varepsilon|_{\varepsilon=0}$.
Substituting into $\Pi'(0)=\sum_x G(x)\dot\mu^0(x)$ gives
\begin{equation}
\label{eq:pi_prime_g}
\Pi'(0)
= \langle\dot\theta^0,\,\operatorname{Cov}_{\mu^{\star}}(G,T)\rangle
= \langle\dot\theta^0,\,g_G\rangle,
\end{equation}
where $(g_G)_k = \operatorname{Cov}_{\mu^\star}(G,T_k)$ over all sufficient statistics.

From Theorem~\ref{thm:differentiability}, $\dot\theta^0$ solves
$H\,\dot\theta^0=h$ with $h$ as in~\eqref{eq:perturbation_vector},
so $\dot\theta^0=H^+h$.  Substituting into \eqref{eq:pi_prime_g},
\[
\Pi'(0)
= \langle\dot\theta^0,\,g_G\rangle
= \langle H^+h,\,g_G\rangle
= \langle h,\,H^+g_G\rangle
= \langle h,\,\Psi\rangle,
\]
where $\Psi=H^+g_G$.  Since $h$ has non-zero entries only in the
marginal block, $\Pi'(0)=\langle h_1,\psi_1\rangle+\langle h_V,\psi_V\rangle
+\langle h_2,\psi_2\rangle$ with $\psi$ the marginal components of $\Psi$.
\[
\Pi'(0)
= \langle h_{1},\psi_{1}\rangle
 + \langle h_{V},\psi_{V}\rangle
 + \langle h_{2},\psi_{2}\rangle.
\]
This completes the proof.
\end{proof}

\begin{remark}
Equation \eqref{eq:risk_representation} states: to first order, the price sensitivity of any payoff $G$ under marginal shocks
is obtained by pairing the marginal shock directions with a set of influence functions computed from the calibrated Gibbs coupling.
This converts ``bump-and-revalue'' into a mathematically controlled linear response.
\end{remark}
\subsection{Second-Order Error Bound}
\label{subsec:second_order}

The linear response formula derived in Theorem~\ref{thm:risk_representation}
is a first-order approximation. The following corollary quantifies its
accuracy, establishing that the approximation error is $O(\varepsilon^2)$
in the perturbation size.

\begin{corollary}[Second-order error bound]
\label{cor:second_order}
Under the assumptions of Theorem~\ref{thm:differentiability}, let
$G:\mathcal{X}\to\mathbb{R}$ be any payoff.
Then there exists $C<\infty$ such that for all sufficiently small
$|\varepsilon|$,
\begin{equation}
\label{eq:second_order_bound}
\bigl|\Pi(\varepsilon)-\Pi(0)-\varepsilon\langle h,\psi\rangle\bigr|
\;\le\; C\varepsilon^2\,\|h\|^2.
\end{equation}
\end{corollary}
\begin{proof}
By Theorem~\ref{thm:differentiability}, $\varepsilon\mapsto\theta^\varepsilon$
is $C^1$; since $\Lambda$ is $C^\infty$, the implicit function theorem gives
$\varepsilon\mapsto\theta^\varepsilon$ is in fact $C^\infty$, and hence
$\varepsilon\mapsto\mu^\varepsilon$ is $C^\infty$ entry-wise.
Therefore $\Pi(\varepsilon)=\sum_x G(x)\mu^\varepsilon(x)$ is $C^2$ in
$\varepsilon$, and Taylor's theorem gives
\[
\Pi(\varepsilon) = \Pi(0) + \varepsilon\Pi'(0) + \tfrac{1}{2}\varepsilon^2\Pi''(\xi)
\]
for some $\xi$ between $0$ and $\varepsilon$, with $\Pi'(0)=\langle h,\psi\rangle$
by Theorem~\ref{thm:risk_representation}.
It remains to show $|\Pi''(\xi)|\le C\|h\|^2$ for a constant $C$ independent of $h$.
From the exponential family differentiation of $\mu^\varepsilon$, the second
derivative $\ddot\mu^\xi$ consists of terms that are either quadratic in
$\dot\theta^\xi$ or linear in $\ddot\theta^\xi$.
By Theorem~\ref{thm:differentiability}, $\dot\theta^\xi = H^{+,\xi}h$
with $\|h\|=\|(h_1,h_V,h_2,\mathbf{0})\|=\|h\|$, so
$\|\dot\theta^\xi\|\le\|H^{+,\xi}\|\|h\|$.
The second-order term satisfies
$\ddot\theta^\xi = -(H^{\xi})^+\nabla^3\Lambda(\theta^\xi)[\dot\theta^\xi,\dot\theta^\xi]$,
which scales as $\|\dot\theta^\xi\|^2$.
Hence every term in $\ddot\mu^\xi$ is $O(\|h\|^2)$.
Since $\mathcal{X}$ is finite and $\mu^\xi>0$ on the compact interval
$|\xi|\le\varepsilon_0$, the operator norms $\|H^{+,\xi}\|$,
$\|\nabla^3\Lambda(\theta^\xi)\|$ are bounded uniformly in $\xi$,
giving $|\Pi''(\xi)|=|\sum_x G(x)\ddot\mu^\xi(x)|\le C\|h\|^2$
for a constant $C$ depending on $\mu^\star$, $\bar\mu$, $G$ but not on $h$.
Substituting into the Taylor remainder completes the proof.
\end{proof}

\begin{remark}
The bound~\eqref{eq:second_order_bound} confirms that for the small
marginal perturbations arising from realistic market bumps, the linear
response formula $\Pi(0)+\varepsilon\langle h,\psi\rangle$ is highly
accurate: the error is quadratic in both the perturbation size $\varepsilon$
and the shock magnitude $\|h\|$.
\end{remark}

\subsection{Stability Bounds}

\begin{proposition}[Lipschitz stability of potentials and couplings]
\label{prop:lipschitz}
Under the assumptions of Theorem~\ref{thm:differentiability},
there exists $C>0$ such that for all sufficiently small $\varepsilon$:
\[
\|u^\varepsilon-u^0\| \le C|\varepsilon|\|(h_1,h_V,h_2)\|,
\qquad
\|\mu^\varepsilon-\mu^\star\|_1 \le C|\varepsilon|\|(h_1,h_V,h_2)\|.
\]
\end{proposition}

\begin{proof}
From Theorem~\ref{thm:differentiability}, the full dual variation satisfies
$\dot\theta^0 = H^+h$ with $\|h\|=\|(h_1,h_V,h_2)\|$.
The marginal-potential component $\dot u^\varepsilon$ is the projection of
$\dot\theta^\varepsilon$ onto the marginal block; its norm satisfies
$\|\dot u^\varepsilon\|\le\|\dot\theta^\varepsilon\|\le\|H^{+,\varepsilon}\|\|h\|$.
Since $H^\varepsilon$ varies continuously in $\varepsilon$ and $\mathcal{X}$ is finite,
$\|H^{+,\varepsilon}\|$ is bounded by some $C$ on a compact neighborhood.
Integrating $\dot u^\varepsilon$ over $\varepsilon$ yields the first bound.
The second bound follows from smoothness of the Gibbs map
$\mu^\varepsilon=\mathcal{G}(\theta^\varepsilon)$ and the bounded Jacobian on the same neighborhood.
\end{proof}
\subsection{Information-Geometric Interpretation}

The perturbation theory admits a precise interpretation in information
geometry~\cite{Amari2016}.

\paragraph{The ambient statistical manifold.}
Fix the prior $\bar\mu>0$ and the sufficient statistics $\{T_k\}$
of~\eqref{eq:MOT_stats}.
The \emph{statistical manifold} is the full exponential family
\begin{equation}
\label{eq:stat_manifold}
\mathcal{M}
\;:=\;
\Bigl\{\,
\mu_\theta(x)
= \bar\mu(x)\,\frac{e^{\langle\theta,T(x)\rangle}}{Z(\theta)}
\;:\;
\theta\in\mathbb{R}^K
\,\Bigr\},
\end{equation}
where $K=N_1+N_V+N_2+2N_1N_V$ and $Z(\theta)=\sum_x\bar\mu(x)e^{\langle\theta,T(x)\rangle}$.
Every point $\mu_\theta\in\mathcal{M}$ is a valid coupling on $\mathcal{X}$;
the calibrated solution $\mu^\star=\mu_{\theta^\star}$ is one specific point on
this manifold, determined by the constraint-matching conditions of
Corollary~\ref{cor:dual_optimality}.
The parameter space $\mathbb{R}^K$ serves as a global coordinate chart,
and the tangent space $T_\theta\mathcal{M}$ at any point is identified with
$\mathbb{R}^K$ via the parametrization $\theta$.

\paragraph{The Riemannian metric.}
By Proposition~\ref{prop:fisher_hessian_cov}, $H(\theta)=\nabla^2_\theta\Lambda(\theta)$
is symmetric positive semi-definite; it serves as the Fisher--Rao metric tensor on
$\mathcal{M}$, defining for tangent vectors $u,v\in T_\theta\mathcal{M}\cong\mathbb{R}^K$
the inner product $\langle u,v\rangle_{H(\theta)}:=u^\top H(\theta)\,v$.
Under Assumption~\ref{ass:non_deg}, $H(\theta)\succ 0$ and the metric is
non-degenerate; in general it degenerates on $\ker(H)$.
The metric is intrinsic to $\mathcal{M}$; it depends only on the calibrated
family, not on any payoff or perturbation.

\paragraph{Direction, functional, and Gâteaux derivative.}
Two further objects are needed to compute a sensitivity:
\begin{itemize}
\item \emph{A tangent direction}: a marginal perturbation
      $h=(h_1,h_V,h_2,0,\ldots,0)\in\mathbb{R}^K$ specifies how the constraint
      targets change.  The linear response system $H\dot\theta^0=h$ resolves
      $h$ into the corresponding tangent vector $\dot\theta^0=H^+h$ in
      $T_{\theta^\star}\mathcal{M}$.
\item \emph{A smooth functional}: a payoff $G:\mathcal{X}\to\mathbb{R}$ defines
      the price map $\Pi:\mathcal{M}\to\mathbb{R}$ by
      $\Pi(\mu_\theta):=\mathbb{E}_{\mu_\theta}[G]$.
      Its Euclidean gradient on $\mathcal{M}$ in $\theta$-coordinates is
      $g_G:=\nabla_\theta\Pi(\theta^\star)=\mathrm{Cov}_{\mu^\star}(G,T)$.
\end{itemize}
The linear response system $H\dot\theta^0=h$ maps the constraint
perturbation $h$ to the tangent vector $\dot\theta^0=H^+h\in
T_{\theta^\star}\mathcal{M}$.
The Gâteaux derivative of $\Pi$ at $\mu^\star$ along $\dot\theta^0$ is
\begin{equation}
\label{eq:gateaux}
\Pi'(0)
\;=\;
g_G^\top\dot\theta^0
\;=\;
g_G^\top H^+h
\;=\;
\bigl\langle H^+g_G,\;\dot\theta^0\bigr\rangle_{H},
\end{equation}
where the last equality uses $\langle u,v\rangle_H:=u^\top Hv$, so that
$(H^+g_G)^\top H\dot\theta^0 = g_G^\top\dot\theta^0$.
The vector $H^+g_G$ is the \emph{natural gradient} of $\Pi$ on
$\mathrm{range}(H)$, its Riemannian gradient with respect to the Fisher
metric restricted to $\mathrm{range}(H)$.
Lemma~\ref{lem:range_gG} guarantees $g_G\in\mathrm{range}(H)$, so $H^+g_G$
is well-defined regardless of whether $H$ is invertible.
Equation~\eqref{eq:gateaux} then states that the sensitivity equals the
Fisher-metric inner product of the natural gradient with $\dot\theta^0$.

This makes the structure explicit: $H$ is fixed at calibration and encodes
the local geometry of $\mathcal{M}$; $G$ and $h$ are chosen at risk-generation
time and specify which sensitivity is evaluated.
\paragraph{Economic interpretation.}

The perturbation framework can be interpreted as solving a nearby optimal
transport problem whose prior is the calibrated coupling $\mu^\star$.
A marginal perturbation corresponds to a change in market forwards or option prices.
The entropic projection identifies the closest joint distribution consistent
with the new market information.

This perspective shows that the risk sensitivities derived in this paper are
not tied to any specific stochastic volatility model.
Instead they arise from the geometry of the calibrated transport plan.
In this sense the risk generation mechanism is largely model independent.
\section{Financial Perturbations: Spot And Volatility Bumps}
\label{sec:financial_perturbations}

This section translates two standard market perturbations (an SPX spot bump
and a parallel SPX implied volatility bump) into admissible perturbation
vectors $h=(h_1,h_V,h_2,\mathbf{0})$ in the sense of
Section~\ref{sec:perturbation}.

The SPX components $h_1$ and $h_2$ are derived here directly from the
perturbed marginals.
The VIX component $h_V$ is \emph{not} zero: an SPX perturbation changes
the SPX forward variance, which shifts the VIX future level $F_V$ through
the forward-variance replication identity, and this shift already induces
a change in the VIX marginal $\mu_V$ before any smile dynamics are imposed.
The SSR rule of Section~\ref{subsec:ssr_compatibility} then further propagates
this shift to the full VIX implied volatility surface.
We therefore leave $h_V$ as a \emph{placeholder} in this section, to be
completed in Section~\ref{subsec:ssr_compatibility}.

\subsection{SPX Spot Perturbation}\label{subsec:spot_perturbation}

\paragraph{Setup.}
Under a sticky-strike convention, shifting the spot $S_0\mapsto S_0+\delta_S$
translates the SPX forward distribution at both maturities $T_1$ and $T_2$
by $\delta_S$ (exact under flat-smile/Bachelier dynamics; a first-order
approximation in general).
On the discrete grids $\mathcal{S}_j=\{s_j^i\}_{i=1}^{N_j}$,
$j\in\{1,2\}$, define
\begin{equation}
\label{eq:spot_perturbed_marginal}
\mu_j^{\delta_S}(s_j^i) \;:=\; \mu_j(s_j^i-\delta_S),
\end{equation}
where $\mu_j(\cdot)$ is extended off-grid by piecewise-linear interpolation.

\begin{assumption}[Interior support]\label{ass:interior_support}
The grids $\mathcal{S}_1$ and $\mathcal{S}_2$ are padded so that
$\mu_j(s_j^1)=\mu_j(s_j^{N_j})=0$ for $j\in\{1,2\}$, i.e., the support of
each marginal is strictly interior to the grid.
\end{assumption}

\begin{remark}
Translation invariance of a Lebesgue integral, $\int f(x-\delta_S)\,dx =
\int f(x)\,dx$, does \emph{not} transfer to a finite discrete sum on a
truncated grid: $\sum_i \mu_j(s_j^i-\delta_S)$ loses mass at the boundary
unless Assumption~\ref{ass:interior_support} holds.
In practice, the grids used in calibration extend well beyond the range of
quoted strikes, so boundary masses are negligible; Assumption~\ref{ass:interior_support}
formalises this.
\end{remark}

\begin{proposition}[Spot bump defines an admissible perturbation]
\label{prop:spot_admissible}
Under Assumption~\ref{ass:interior_support}, for $|\delta_S|<h_{\mathrm{grid}}:=\min_i(s_j^{i+1}-s_j^i)$,
$\mu_j^{\delta_S}$ is a valid probability measure on $\mathcal{S}_j$ for
$j\in\{1,2\}$, and the directional derivatives
\begin{equation}
\label{eq:h1_spot}
h_j(s_j^i) \;:=\;
\left.\frac{d}{d\delta_S}\mu_j^{\delta_S}(s_j^i)\right|_{\delta_S=0}
\;=\; -\mu_j'(s_j^i)
\end{equation}
satisfy $\sum_i h_j(s_j^i)=0$.
\end{proposition}

\begin{proof}
By piecewise-linear interpolation with mesh spacing $h_{\mathrm{grid}}$, for $\delta_S\in(0,h_{\mathrm{grid}})$,
\[
\mu_j^{\delta_S}(s_j^i)
\;=\; \bigl(1-\tfrac{\delta_S}{h_{\mathrm{grid}}}\bigr)\mu_j(s_j^i)
     + \tfrac{\delta_S}{h_{\mathrm{grid}}}\mu_j(s_j^{i-1}),
\]
where $\mu_j(s_j^0):=0$ by convention.
Summing over $i=1,\dots,N_j$ and using Assumption~\ref{ass:interior_support}
($\mu_j(s_j^1)=\mu_j(s_j^{N_j})=0$):
\[
\sum_i\mu_j^{\delta_S}(s_j^i)
\;=\; \bigl(1-\tfrac{\delta_S}{h_{\mathrm{grid}}}\bigr)\sum_i\mu_j(s_j^i)
     + \tfrac{\delta_S}{h_{\mathrm{grid}}}\sum_i\mu_j(s_j^{i-1})
\;=\; 1,
\]
where the last equality holds because both inner sums equal 1 (the
boundary-zero assumption prevents mass leaving the grid).
Positivity for small $|\delta_S|$ follows from linear interpolation.
Differentiating~\eqref{eq:spot_perturbed_marginal} at $\delta_S=0$ gives
$h_j = -\mu_j'$~\eqref{eq:h1_spot}, and $\sum_i h_j(s_j^i) =
d/d\delta_S\bigl[\sum_i\mu_j^{\delta_S}(s_j^i)\bigr]_{\delta_S=0} = 0$.
\end{proof}

\begin{remark}
An equivalent construction that bypasses Assumption~\ref{ass:interior_support}
is the \emph{barycentric pushforward}: each atom at $s_j^i$ is mapped to
$s_j^i+\delta_S$, then redistributed linearly between the two nearest grid
points (with mass reaching outside the grid absorbed at the boundary).
This defines a stochastic matrix $P_{\delta_S}$ and sets $\mu_j^{\delta_S}:=P_{\delta_S}\mu_j$;
mass is preserved exactly for all $\delta_S$, and the resulting $h_j$ agrees
with $-\mu_j'$ to first order.
\end{remark}

\paragraph{Perturbation vector.}
For the spot bump, the perturbation vector is
\begin{equation}
\label{eq:h_spot}
h^{\mathrm{spot}} \;=\;
\Bigl(\,-\mu_1'(\cdot),\;
      \underbrace{h_V}_{\text{see Sec.~\ref{subsec:ssr_compatibility}}},\;
      -\mu_2'(\cdot),\;
      \underbrace{0,\ldots,0}_{\text{financial slots}}\,\Bigr)
\;\in\;\mathbb{R}^K,
\end{equation}
where $-\mu_j'$ denotes the discrete derivative of the marginal
$\mu_j$~\eqref{eq:h1_spot}.
The VIX component $h_V$ is non-zero because the spot bump shifts the SPX
forward variance and hence the VIX future $F_V$; its explicit form is
derived in Section~\ref{subsec:ssr_compatibility}.

\subsection{SPX Implied Volatility Perturbation}\label{subsec:vol_perturbation}

\paragraph{Setup.}
A parallel implied volatility bump applied to both SPX maturities shifts
\[
\sigma(K,T_i) \;\mapsto\; \sigma(K,T_i)+\delta_\sigma, \qquad i\in\{1,2\}.
\]
At each maturity $T_i$, the perturbed call prices follow from Black-Scholes~\cite{black1973pricing},
\begin{equation}
\label{eq:bumped_calls}
C_i^{\delta_\sigma}(K) \;:=\; C_{BS}\!\bigl(S_0,K,T_i,\,\sigma(K,T_i)+\delta_\sigma\bigr),
\end{equation}
and Breeden-Litzenberger inversion~\cite{breeden1978prices} yields the perturbed marginal density:
\begin{equation}
\label{eq:BL_perturbed}
\mu_i^{\delta_\sigma}(s) \;=\; \frac{\partial^2 C_i^{\delta_\sigma}}{\partial K^2}\bigg|_{K=s},
\qquad i\in\{1,2\}.
\end{equation}

\begin{assumption}[Strict positivity of base marginals on the grid]
\label{ass:non_deg_marg}
The calibrated base marginals satisfy $\mu_i(s)>0$ for every grid point
$s\in\mathcal{S}_i$, $i\in\{1,2\}$.
\end{assumption}

\begin{remark}
Assumption~\ref{ass:non_deg_marg} is a consequence of the strictly positive
feasible coupling (A3) of Theorem~\ref{thm:existence}: if
$\mu_0\in\mathcal{P}$ with $\mu_0(s_1,v,s_2)>0$ everywhere, then each
marginal $\mu_i(s)=\sum_{(s_1,v,s_2):s_i=s}\mu_0(s_1,v,s_2)>0$.
Economically, it states that the calibrated model assigns positive
probability to every grid state, a property satisfied by any
Gibbs coupling with a strictly positive prior $\bar\mu>0$.
\end{remark}

\begin{proposition}[Volatility bump defines admissible perturbations]
\label{prop:vol_admissible}
For sufficiently small $|\delta_\sigma|$, $\mu_i^{\delta_\sigma}$ is a valid probability
measure on $\mathcal{S}_i$ for each $i\in\{1,2\}$, and the directional
derivatives
\begin{equation}
\label{eq:hi_vol}
h_i(s) \;:=\;
\left.\frac{d}{d\delta_\sigma}\mu_i^{\delta_\sigma}(s)\right|_{\delta_\sigma=0}
\;=\; \frac{\partial^2\vegaSPX[,i]}{\partial K^2}\bigg|_{K=s},
\end{equation}
where $\vegaSPX[,i](K):=\partial_\sigma C_{BS}(S_0,K,T_i,\sigma(K,T_i))$
is the Black-Scholes vega at maturity $T_i$ and strike $K$, each satisfy
$\sum_s h_i(s)=0$.
\end{proposition}

\begin{proof}
Since $C_{BS}$ is $C^\infty$ in $\sigma$, the map $\delta_\sigma\mapsto C_i^{\delta_\sigma}(K)$
is $C^\infty$ jointly in $(K,\delta_\sigma)$.
The Breeden-Litzenberger density therefore expands as
\begin{equation}
\label{eq:density_expansion}
\mu_i^{\delta_\sigma}(s)
\;=\; \mu_i(s) + \delta_\sigma\,h_i(s) + O(\delta_\sigma^2),
\end{equation}
where $h_i(s)=\partial^2_K\vegaSPX[,i]|_{K=s}$ as in~\eqref{eq:hi_vol}.
The naive implication $\sigma+\delta_\sigma>0\Rightarrow\partial^2_K C^{\delta_\sigma}\geq 0$
fails on a skewed smile (the Breeden-Litzenberger formula involves $\sigma'(K)$ and
$\sigma''(K)$, so density positivity is not implied merely by vol positivity).
The correct argument is quantitative and discrete.
By Assumption~\ref{ass:non_deg_marg}, the base marginal is strictly
positive on the finite grid $\mathcal{S}_i$, so
\[
\underline\mu_i \;:=\; \min_{s\in\mathcal{S}_i}\mu_i(s) \;>\; 0,
\qquad
\overline{h}_i \;:=\; \max_{s\in\mathcal{S}_i}|h_i(s)| \;<\; \infty.
\]
Setting $\bar\delta_{\sigma,i} := \underline\mu_i/(2\overline{h}_i)$, uniform
continuity of the $O(\delta_\sigma^2)$ remainder on the finite grid gives a
constant $C_i$ such that for $|\delta_\sigma|\le\bar\delta_{\sigma,i}$:
\[
\mu_i^{\delta_\sigma}(s) \;\ge\; \underline\mu_i - |\delta_\sigma|\,\overline{h}_i - C_i\delta_\sigma^2
\;\ge\; \tfrac{1}{2}\underline\mu_i \;>\; 0.
\]
Additionally $|\delta_\sigma|<\min_K\sigma(K,T_i)$ ensures $\sigma(K,T_i)+\delta_\sigma>0$
(implied vols remain positive), so $C_i^{\delta_\sigma}$ is a valid Black-Scholes price.
The admissible range is $|\delta_\sigma|\le\min\{\bar\delta_{\sigma,i},\min_K\sigma(K,T_i)\}$.

Mass preservation follows from the boundary conditions $C_i^{\delta_\sigma}(0)=S_0$
and $C_i^{\delta_\sigma}(K)\to 0$ as $K\to\infty$, which hold for all $\delta_\sigma$.
Differentiating~\eqref{eq:BL_perturbed} with respect to $\delta_\sigma$ at $\delta_\sigma=0$
and exchanging $d/d\delta_\sigma$ with $\partial^2_K$ (justified by smoothness):
\[
h_i(s)
= \frac{\partial^2}{\partial K^2}
  \underbrace{\frac{d C_i^{\delta_\sigma}}{d\delta_\sigma}\bigg|_{\delta_\sigma=0}}_{\vegaSPX[,i](K)}
  \bigg|_{K=s}
= \frac{\partial^2\vegaSPX[,i]}{\partial K^2}\bigg|_{K=s},
\]
establishing~\eqref{eq:hi_vol}.
The mass-zero conditions follow by differentiating $\sum_s\mu_i^{\delta_\sigma}(s)=1$
with respect to $\delta_\sigma$ at $\delta_\sigma=0$, as in
Proposition~\ref{prop:spot_admissible}.
\end{proof}

\paragraph{Perturbation vector.}
For the parallel vol bump, the perturbation vector is
\begin{equation}
\label{eq:h_vol}
h^{\mathrm{vol}} \;=\;
\Bigl(\,\partial^2_K\vegaSPX[,T_1]\big|_{K=\cdot},\;
      \underbrace{h_V}_{\text{see Sec.~\ref{subsec:ssr_compatibility}}},\;
      \partial^2_K\vegaSPX[,T_2]\big|_{K=\cdot},\;
      \underbrace{0,\ldots,0}_{\text{financial slots}}\,\Bigr)
\;\in\;\mathbb{R}^K,
\end{equation}
where $\vegaSPX[,T_i](K)=\partial_\sigma C_{BS}(S_0,K,T_i,\sigma(K,T_i))$
is the Black-Scholes vega at maturity $T_i$~\eqref{eq:hi_vol}.
If instead a single-maturity bump is applied (e.g.\ only $T_1$), then the
$\partial^2_K\vegaSPX[,T_2]$ component is zero.
In both cases $h_V\neq 0$ and is completed in
Section~\ref{subsec:ssr_compatibility}.

\paragraph{Summary.}
In both cases the perturbation vector $h$ is admissible in the sense of
Section~\ref{sec:perturbation}: each non-zero component is mass-preserving
and the financial constraint slots are zero.
The risk sensitivity for any payoff $G$ is then
$\Pi'(0)=g_G^\top H^+h$, with $h$ assembled from~\eqref{eq:h_spot}
or~\eqref{eq:h_vol} respectively.
The LR method enforces all three components $(h_1,h_V,h_2)$ simultaneously
through the Fisher linear system; the treatment of $h_2$ under the DR method
is discussed in Remark~\ref{rem:S2_endogenous} once that framework has been
introduced (Section~\ref{sec:dimension_reduction}).

\section{SSR Dynamics For VIX And Its Linearization}
\subsection{Skew Stickiness Ratio: Bergomi's Definition And Extension To VIX}

We now give a formal definition of the Skew Stickiness Ratio (SSR) following
~\cite{bergomi2016,bergomi2009Slides}.  
Let $\sigma(K,F)$ denote the implied volatility of an option with strike $K$
and forward level $F$ of the underlying asset.  
Bergomi models smile dynamics by expressing first-order reactions of the
volatility surface to changes in the forward.

\subsubsection{Bergomi’s Definition of Skew Stickiness Ratio}

Following \cite{bergomi2016,bergomi2009Slides}, let $\sigma^*(F):=\sigma(F,F)$ denote the
at-the-money-forward (ATMF) implied volatility and let
\begin{equation}
\label{eq:Slog}
S_{\log}(F)
:= \left.\frac{\partial\sigma}{\partial\ln K}\right|_{K=F}
 = F\left.\frac{\partial\sigma(K,F)}{\partial K}\right|_{K=F}
\end{equation}
be the \emph{ATM log-moneyness skew}.  For equity underlyings the leverage
effect drives $S_{\log}<0$.

Bergomi's Skew Stickiness Ratio $\beta$ (SSR) is the dimensionless parameter
defined through the floating-ATMF sensitivity
\begin{equation}
\label{eq:BergomiSSROriginal}
\frac{d\sigma^*(F)}{d\ln F} = \beta\cdot S_{\log}(F).
\end{equation}
Equivalently, the sensitivity of fixed-strike vol at $K=F$ is
\begin{equation}
\label{eq:BergomiFixedStrike}
\left.\frac{\partial\sigma(K,F)}{\partial F}\right|_{K=F}
= (\beta-1)\,\frac{S_{\log}(F)}{F}
= (\beta-1)\left.\frac{\partial\sigma}{\partial K}\right|_{K=F},
\end{equation}
and for a general strike $K$ the self-similar (SSR) surface satisfies
\begin{equation}
\label{eq:BergomiGeneralK}
\left.\frac{\partial\sigma(K,F)}{\partial F}\right|_{K\,\text{fixed}}
= (\beta-1)\,\frac{K}{F}\left.\frac{\partial\sigma}{\partial K}\right|_{K}.
\end{equation}
Thus a perturbation $\delta F$ in the forward produces the leading-order
smile shift
\begin{equation}
\label{eq:BergomiShift}
\delta\sigma(K)
= (\beta-1)\,\frac{K}{F}\,\partial_K\sigma(K)\,\delta F
\qquad (K\text{ fixed}),
\end{equation}
or, using the ATM uniform approximation,
$\delta\sigma \approx (\beta-1)\,S_{\log}(F)\,\delta F/F$.

The stickiness regimes are:
\begin{itemize}
\item $\beta=0$: Sticky delta (smile frozen in $K/F$; $d\sigma^*/d\ln F=0$),
\item $\beta=1$: Sticky strike (fixed-$K$ vol frozen; $\partial_F\sigma|_K=0$),
\item $\beta\approx\tfrac{3}{2}$: Empirical SPX value \cite{bergomi2016},
\item $\beta>1$: Super-skew (fixed-$K$ vol moves \emph{opposite} to $F$, amplifying the leverage effect beyond sticky-strike).
\end{itemize}

\begin{remark}[Sign convention for equity]
For equity indices $S_{\log}<0$ and $\beta>0$, so $d\sigma^*/d\ln F = \beta S_{\log}<0$:
the floating ATMF vol rises when the spot falls.  For $\beta>1$ the fixed-$K$
vol also falls when the spot rises ($\partial_F\sigma|_{K=F}=(\beta-1)S_{\log}/F<0$),
consistent with the leverage effect.  At $\beta=1$ (sticky-strike) the fixed-$K$
vol is unchanged, and the entire ATMF rise comes from the smile sliding along
the strike axis.
\end{remark}

\subsubsection{Extension to VIX Futures And VIX Options}

Let $F_V$ denote the VIX future level, $\sigma_V(K,F_V)$ the VIX option
implied volatility, and $\sigma_V^*(F_V):=\sigma_V(F_V,F_V)$ the ATMF VIX vol.
Define the VIX ATM log-moneyness skew
$$
S_{\log,V}(F_V)
:= F_V\left.\frac{\partial\sigma_V(K,F_V)}{\partial K}\right|_{K=F_V}.
$$
VIX options have a positively-sloped smile, so $S_{\log,V}>0$.

The VIX Skew Stickiness Ratio $\beta_V$ is defined by the same floating-ATMF
condition as \eqref{eq:BergomiSSROriginal}:
\begin{equation}
\label{eq:VIXSSRDefinition}
\frac{d\sigma_V^*(F_V)}{d\ln F_V} = \beta_V\cdot S_{\log,V}(F_V).
\end{equation}
The corresponding fixed-strike sensitivity for general $K$ is
\begin{equation}
\label{eq:VIXFixedStrikeGeneral}
\left.\frac{\partial\sigma_V(K,F_V)}{\partial F_V}\right|_{K\,\text{fixed}}
= (\beta_V-1)\,\frac{K}{F_V}\,\partial_K\sigma_V(K),
\qquad \beta_V > 0.
\end{equation}
At the ATM point $K=F_V$ this reduces to
\begin{equation}
\label{eq:VIXFixedStrike}
\left.\frac{\partial\sigma_V(K,F_V)}{\partial F_V}\right|_{K=F_V}
= (\beta_V-1)\,\frac{S_{\log,V}}{F_V}
= (\beta_V-1)\left.\frac{\partial\sigma_V}{\partial K}\right|_{K=F_V}.
\end{equation}
The corresponding linearized vol shift for a finite move $\delta F_V$ is
\begin{equation}
\label{eq:VIXLinearVolShift}
\delta\sigma_V(K)
= (\beta_V-1)\,\frac{K}{F_V}\,\partial_K\sigma_V(K)\,\delta F_V
\qquad (K\text{ fixed}).
\end{equation}

\begin{remark}[Sign convention for VIX]
Since $S_{\log,V}>0$ and empirically $\beta_V\approx 1.2>1$, we have
$(\beta_V-1)>0$: when VIX rises ($\delta F_V>0$) the fixed-$K$ VIX vol also
rises ($\delta\sigma_V>0$), consistent with market observation.  At $\beta_V=1$
(sticky-strike) the VIX vol surface would be frozen in absolute strike space.
\end{remark}

\subsection{Linear SSR Approximation And Second--Order Accuracy}

Using the SSR definition~\eqref{eq:VIXSSRDefinition} and the fixed-strike
identity~\eqref{eq:VIXFixedStrikeGeneral}, we now formalise the linearisation
with a quantitative second-order error bound.

\begin{theorem}[Unified linear SSR expansion with second-order error]
\label{thm:SSRUnified}
Assume $\sigma_V(K,F_V)$ is $C^2$ in $F_V$ near the base level.
Let $F_V' = F_V + \delta F_V$.  Then the implied volatility satisfies
\begin{equation}
\label{eq:SSRUnifiedExpansion}
\sigma_V(K,F_V')
=
\sigma_V(K,F_V)
+(\beta_V-1)\,\frac{K}{F_V}\,\partial_K\sigma_V(K)\,\delta F_V
+R_K(\delta F_V),
\end{equation}
where the remainder satisfies the deterministic bound
\begin{equation}
\label{eq:SSRUnifiedBound}
\bigl|R_K(\delta F_V)\bigr|
\le
\frac{1}{2}
\sup_{\lvert u-F_V\rvert\le |\delta F_V|}
\left|
\frac{\partial^2\sigma_V(K,u)}{\partial F_V^2}
\right|
(\delta F_V)^2.
\end{equation}

In particular, the \emph{strike-dependent linear SSR approximation}
\begin{equation}
\label{eq:SSRLinearApprox}
\sigma_V(K,F_V')
\approx
\sigma_V(K,F_V)
+(\beta_V-1)\,\frac{K}{F_V}\,\partial_K\sigma_V(K)\,(F_V'-F_V)
\end{equation}
is first-order accurate with $O((\delta F_V)^2)$ error controlled by
\eqref{eq:SSRUnifiedBound}.  The \emph{uniform (ATM) approximation} replaces
$\frac{K}{F_V}\partial_K\sigma_V(K)$ by its ATM value $S_{\log,V}/F_V$:
\begin{equation}
\label{eq:SSRLinearATM}
\sigma_V(K,F_V')
\approx
\sigma_V(K,F_V)
+(\beta_V-1)\,\frac{S_{\log,V}}{F_V}\,(F_V'-F_V).
\end{equation}
\end{theorem}

\begin{proof}
Apply Taylor's theorem in $F_V$:
\[
\sigma_V(K,F_V+\delta F_V)
=\sigma_V(K,F_V)
+\frac{\partial\sigma_V}{\partial F_V}(K,F_V)\,\delta F_V
+\frac{1}{2}\,\frac{\partial^2\sigma_V}{\partial F_V^2}(K,\xi)\,(\delta F_V)^2,
\]
for some $\xi$ between $F_V$ and $F_V+\delta F_V$.  Substituting the SSR
identity~\eqref{eq:BergomiGeneralK}
\[
\frac{\partial\sigma_V}{\partial F_V}(K,F_V)
= (\beta_V-1)\,\frac{K}{F_V}\,\partial_K\sigma_V(K)
\]
into the first-order term gives \eqref{eq:SSRUnifiedExpansion}, and the
bound \eqref{eq:SSRUnifiedBound} follows by taking the supremum over the
interval. \qedhere
\end{proof}

\paragraph{Interpretation.}
The first-order term in \eqref{eq:SSRUnifiedExpansion} gives the SSR-based
linear smile reaction used in the perturbed optimal transport constraints.
The explicit $O(\delta^2)$ bound in \eqref{eq:SSRUnifiedBound} quantifies the
approximation error and shows that the SSR mapping is highly accurate for the
small forward perturbations relevant for risk generation.

\subsection{Empirical Evidence For VIX Skew Stickiness Ratio}

Skew Stickiness Ratio (SSR) is widely used by practitioners to describe how the VIX
volatility smile responds to changes in the underlying SPX level.

To estimate SSR empirically we regress daily changes in VIX implied volatility
against changes in the VIX future level across different maturities.
Figure~\ref{fig:vix_ssr} illustrates the estimated SSR term
structure using historical windows from six months to three years.

\begin{figure}[H]
\centering

\begin{subfigure}{0.48\textwidth}
  \centering
  \includegraphics[width=\linewidth]{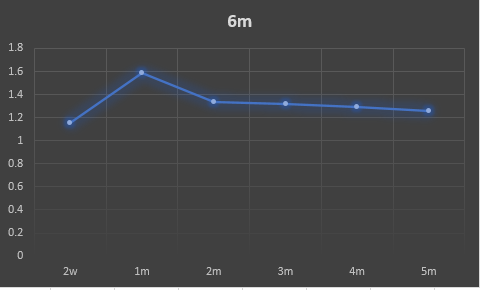}
  \caption{6m window}
  \label{fig:vix_ssr_6m}
\end{subfigure}\hfill
\begin{subfigure}{0.48\textwidth}
  \centering
  \includegraphics[width=\linewidth]{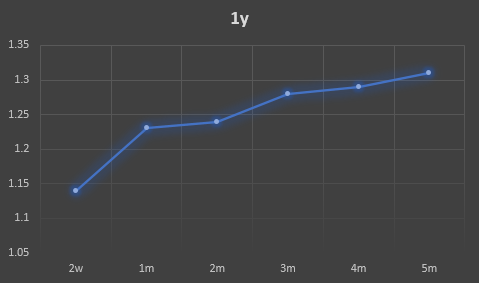}
  \caption{1y window}
  \label{fig:vix_ssr_1y}
\end{subfigure}

\medskip

\begin{subfigure}{0.48\textwidth}
  \centering
  \includegraphics[width=\linewidth]{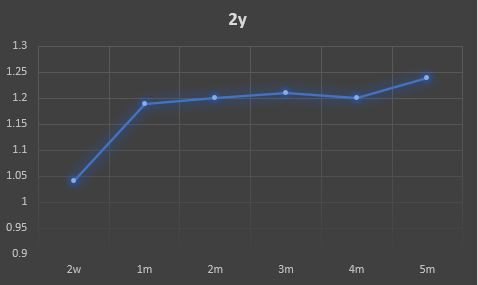}
  \caption{2y window}
  \label{fig:vix_ssr_2y}
\end{subfigure}\hfill
\begin{subfigure}{0.48\textwidth}
  \centering
  \includegraphics[width=\linewidth]{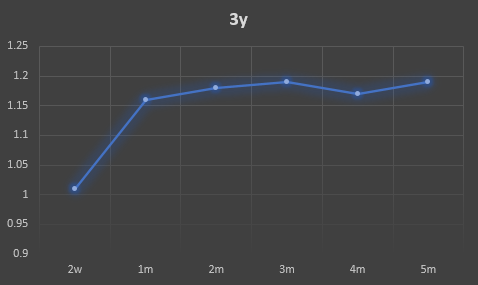}
  \caption{3y window}
  \label{fig:vix_ssr_3y}
\end{subfigure}

\caption{Estimated VIX SSR term structure across four historical windows.}
\label{fig:vix_ssr}
\end{figure}

These estimates confirm $\beta_V\approx 1.2$ used in the numerical experiments.

\subsection{Completing \texorpdfstring{$h_V$}{hV}: SSR Propagation To The VIX Marginal}
\label{subsec:ssr_compatibility}

We now complete the perturbation vector left open in
Section~\ref{sec:financial_perturbations} by deriving the VIX marginal
component $h_V$.
The argument has four steps: (1) compute the forward-variance shift
$\delta\sigma^2_\mathrm{fwd}$ from the SPX marginal perturbation;
(2) apply the linearized SSR rule to define the sensitivity kernel
$\Gamma(K)$ and the shape vector $A_\mathrm{SSR}$;
(3) determine $\delta F_V$ exactly from the variance-consistency constraint;
(4) invert via Breeden--Litzenberger to obtain $h_V$.

\paragraph{Step 1: Forward-variance shift from the SPX marginals.}
The SPX-implied forward variance is the expected log-contract:
\[
\sigma^2_{\mathrm{fwd}}
\;=\; \mathbb{E}_{\mu^\star}[L(S_2/S_1)].
\]
Since $L(S_2/S_1) = L(S_2) - L(S_1)$ (both $L(s)=-\tfrac{2}{\tau}\ln s$ applied
to absolute levels), the expectation separates into marginals:
\[
\sigma^2_{\mathrm{fwd}}
\;=\; \mathbb{E}_{\mu_2}[L(\cdot)] - \mathbb{E}_{\mu_1}[L(\cdot)]
\;=\; \sum_{s_2\in\mathcal{S}_2} L(s_2)\,\mu_2(s_2)
     \;-\; \sum_{s_1\in\mathcal{S}_1} L(s_1)\,\mu_1(s_1).
\]
Under the SPX marginal perturbation $(h_1,h_2)$, the first-order change
in forward variance is
\begin{equation}
\label{eq:deltaFwdVar}
\delta\sigma^2_{\mathrm{fwd}}
\;=\;
\langle L,\,h_2\rangle_{\mathcal{S}_2}
\;-\;
\langle L,\,h_1\rangle_{\mathcal{S}_1},
\end{equation}
where $\langle L,h_i\rangle_{\mathcal{S}_i}=\sum_{s\in\mathcal{S}_i}L(s)\,h_i(s)$.
This is exact (by linearity of expectation) and depends only on the SPX
marginal components $h_1$, $h_2$.

\paragraph{Step 2: SSR sensitivity kernel.}
Applying the SSR rule~\eqref{eq:VIXLinearVolShift}, the first-order
change in VIX implied volatility is $\delta\sigma_V(K)$ as defined there.
The corresponding change in VIX call prices via Black's formula is:
\begin{equation}
\label{eq:deltaCV}
\delta C_V(K)
\;=\;
\underbrace{\Bigl[\Delta^{BS}_V(K)
+(\beta_V-1)\,\frac{K}{F_V}\,\partial_K\sigma_V(K)\,\vegaVIX(K)
\Bigr]}_{\displaystyle\Gamma(K)}\,\delta F_V,
\end{equation}
where $\Delta^{BS}_V(K)=\partial_{F_V}C_V(K)$ is the VIX option delta,
$\vegaVIX(K)=\partial_{\sigma_V}C_V(K)$ is the VIX option vega,
and $\Gamma(K)$ is the combined sensitivity kernel.
Applying Breeden-Litzenberger~\cite{breeden1978prices} to~\eqref{eq:deltaCV} gives the shape
vector
\begin{equation}
\label{eq:A_SSR}
A_{\mathrm{SSR}}(v) \;:=\; \frac{\partial^2\Gamma}{\partial K^2}\bigg|_{K=v},
\end{equation}
where $v\in\mathcal{V}$ denotes a generic VIX grid point,
so that $h_V = A_{\mathrm{SSR}}\,\delta F_V$ once $\delta F_V$ is known.

\begin{lemma}[Positivity of the SSR denominator]
\label{lem:ssr_denom}
Assume:
\begin{enumerate}
\item[(B1)] \emph{(Boundary decay)} The VIX call surface satisfies
            $(K-F_V)^2|\Gamma'(K)|\to 0$ as $K\to 0^+$ and $K\to\infty$.
\item[(B2)] \emph{(Non-negative kernel)} $\Gamma(K)\geq 0$ for all $K>0$.
\end{enumerate}
Then $\langle v^2, A_{\mathrm{SSR}}\rangle > 0$.
\end{lemma}
\begin{proof}
Integrate by parts twice with $u=K^2$ and $dv=\Gamma''(K)\,dK$:
\[
\int_0^\infty K^2\,\Gamma''(K)\,dK
\;=\;
\bigl[K^2\Gamma'(K)\bigr]_0^\infty
- 2\bigl[K\Gamma(K)\bigr]_0^\infty
+ 2\int_0^\infty\Gamma(K)\,dK.
\]
Both boundary terms vanish.
The first, $[K^2\Gamma'(K)]_0^\infty$, vanishes by (B1): since $K^2/(K-F_V)^2\to 1$
as $K\to\infty$ and $K^2/F_V^2\to 0$ as $K\to 0^+$, (B1) implies
$K^2|\Gamma'(K)|\to 0$ at both endpoints.
The second, $[2K\Gamma(K)]_0^\infty$, vanishes trivially at $K\to 0^+$
(since $K\to 0$ with $\Gamma(0)$ finite) and exponentially at $K\to\infty$
(since $\Gamma(K)\leq\Delta^{BS}_V(K)$ decays faster than any polynomial).
Hence
\[
\langle v^2,A_{\mathrm{SSR}}\rangle
\;=\; 2\int_0^\infty\Gamma(K)\,dK.
\]
Under (B2), $\Gamma\geq 0$ on $(0,\infty)$.
Since $\Gamma(0)=\Delta^{BS}_V(0)=1>0$ and $\Gamma$ is continuous,
$\int_0^\infty\Gamma(K)\,dK>0$, so $\langle v^2,A_{\mathrm{SSR}}\rangle>0$.
\end{proof}

\begin{remark}[When (B2) holds]
\label{rem:B2}
Both terms of $\Gamma(K)=\Delta^{BS}_V(K)+(\beta_V-1)\frac{K}{F_V}
\partial_K\sigma_V(K)\vegaVIX(K)$ are non-negative when
$\beta_V\geq 1$ and the VIX skew is non-negative ($\partial_K\sigma_V\geq 0$),
which is the empirically observed regime for VIX ($\beta_V\approx 1.2$,
positive skew).
At $\beta_V=1$ (sticky-strike), $\Gamma=\Delta^{BS}_V>0$ always.
For $\beta_V<1$, (B2) requires that the delta term dominates the
(now negative) vega contribution, which holds whenever $|1-\beta_V|$ is
not too large relative to the skew level.
\end{remark}

\paragraph{Step 3: Variance-anchored determination of $\delta F_V$.}
The MOT variance-consistency condition $\mathbb{E}_{\mu^\star}[V^2]=\sigma^2_{\mathrm{fwd}}$
links the second moment of $\mu_V$ to the forward variance.
Differentiating the exact identity
$\sigma^2_{\mathrm{fwd}} = F_V^2 + \mathrm{Var}_{\mu_V}(V)$
and substituting $h_V = A_{\mathrm{SSR}}\,\delta F_V$:
\begin{align*}
\delta\sigma^2_{\mathrm{fwd}}
&\;=\; 2F_V\,\delta F_V + \delta\mathrm{Var}_{\mu_V}(V) \\
&\;=\; 2F_V\,\delta F_V
  + \bigl(\langle v^2, A_{\mathrm{SSR}}\rangle - 2F_V\bigr)\delta F_V \\
&\;=\; \langle v^2,\,A_{\mathrm{SSR}}\rangle\,\delta F_V,
\end{align*}
where the second line uses
$\delta\mathrm{Var}_{\mu_V}(V) = \langle v^2, h_V\rangle - 2F_V\,\delta F_V
= \langle v^2, A_{\mathrm{SSR}}\rangle\,\delta F_V - 2F_V\,\delta F_V$.
By Lemma~\ref{lem:ssr_denom} (under (B1)--(B2)), the denominator
$\langle v^2,A_{\mathrm{SSR}}\rangle>0$, so solving gives the exact,
non-circular formula:
\begin{equation}
\label{eq:deltaFV}
\delta F_V
\;=\;
\frac{\delta\sigma^2_{\mathrm{fwd}}}{\langle v^2,\,A_{\mathrm{SSR}}\rangle_{\mathcal{V}}}.
\end{equation}

\paragraph{Step 4: Breeden--Litzenberger inversion.}
With $\delta F_V$ determined by~\eqref{eq:deltaFV}, the VIX marginal
perturbation is:
\begin{equation}
\label{eq:hV}
h_V(v)
\;=\; A_{\mathrm{SSR}}(v)\,\delta F_V
\;=\;
\frac{\partial^2\Gamma}{\partial K^2}\bigg|_{K=v}
\cdot
\frac{\delta\sigma^2_{\mathrm{fwd}}}{\langle v^2,\,A_{\mathrm{SSR}}\rangle},
\end{equation}
where $A_{\mathrm{SSR}}$ and $\Gamma(K)$ are defined in~\eqref{eq:A_SSR}
and~\eqref{eq:deltaCV}.
The mass-zero condition $\sum_v h_V(v)=0$ follows because
$\langle 1, A_{\mathrm{SSR}}\rangle = 0$.
By the fundamental theorem of calculus,
$\int_0^\infty \Gamma''(K)\,dK = \bigl[\Gamma'(K)\bigr]_0^\infty$,
where $\Gamma'(K) = -\partial_{F_V}P(V>K)$ is the derivative of the
(negated) VIX risk-neutral survival function with respect to $F_V$.
At $K\to 0^+$, $P(V>0)=1$ identically (VIX is strictly positive),
so $\Gamma'(0^+)=0$; at $K\to\infty$, $P(V>K)\to 0$ identically,
so $\Gamma'(\infty)=0$.
Hence $\langle 1, A_{\mathrm{SSR}}\rangle = 0 - 0 = 0$.

\paragraph{Complete perturbation vector.}
Assembling all three components, the full admissible perturbation vector is
\begin{equation}
\label{eq:h_augmented}
h
\;=\;
\bigl(\,h_1,\;h_V,\;h_2,\;\underbrace{0,\ldots,0}_{\text{financial slots}}\,\bigr)
\;\in\;\mathbb{R}^K,
\end{equation}
where $h_1$, $h_2$ are given by Section~\ref{sec:financial_perturbations},
$h_V$ by~\eqref{eq:hV}, and the financial slots remain zero since the
martingale and variance-consistency constraints are structural.
The risk sensitivity of any payoff $G$ under the fully propagated
SPX perturbation is then
\[
\Pi'(0) \;=\; g_G^\top H^+ h,
\]
with $h$ from~\eqref{eq:h_augmented}.
    
\section{Algorithm For SPX--VIX Risk Without Recalibration }
\subsection{Base Calibration}
\label{subsec:base}
We first compute a base joint coupling between the SPX state $S_1$, the VIX variable $V$, and the future SPX state $S_2$.
The goal of this calibration step is to construct a probability measure
\[
\mu(s_1,v,s_2)
\]
that matches the prescribed SPX and VIX marginals while simultaneously enforcing the key financial consistency conditions linking SPX and VIX dynamics.

Specifically, the calibrated coupling must satisfy:

\begin{itemize}
\item the marginal constraints
\[
(\mathcal{A}_{\mathrm{marg}}\mu)_1=\mu_1, \qquad (\mathcal{A}_{\mathrm{marg}}\mu)_V=\mu_V, \qquad (\mathcal{A}_{\mathrm{marg}}\mu)_2=\mu_2,
\]
corresponding to the SPX spot, VIX, and future SPX marginals implied by market prices;

\item the financial constraints $\mathcal{A}_{\mathrm{fin}}\mu=0$,
i.e.\ the martingale and variance-consistency conditions
defined in Section~\ref{sec:framework}.
\end{itemize}

Numerically, we solve this constrained calibration problem using a nested scheme.
An outer Sinkhorn iteration enforces the marginal constraints via multiplicative scaling factors, while an inner Newton (or damped Newton) correction enforces the conditional martingale and variance-consistency conditions at each $(S_1,V)$ node.

This structure closely follows the constrained calibration framework introduced in the SPX--VIX joint calibration methodology of ~\cite{guyon2020spxvix}, where optimal transport techniques are combined with financial consistency constraints to produce arbitrage-consistent joint distributions.

The resulting calibrated coupling $\mu^\star$ serves as the base distribution for the subsequent perturbation and risk-generation procedures described in the following sections.

\noindent\textbf{Algorithm 1: Base Calibration Via Sinkhorn With Newton/LM Enforcement}
\label{algo:base}
\begin{enumerate}
  \item \textbf{Inputs.} Discrete grids $\mathcal{S}_1,\mathcal{V},\mathcal{S}_2$; target marginals $\mu_1,\mu_V,\mu_2$;
  prior $\bar\mu(s_1,v,s_2)$; log-return functional $L(\cdot)$; marginal tolerance $\varepsilon_{\mathrm{marg}}$;
  financial tolerance $\varepsilon_{\mathrm{fin}}$; Newton damping $\lambda\ge0$.
  \item \textbf{Outputs.} Calibrated coupling $\mu^\star(s_1,v,s_2)$ and Fisher information matrix $H$.
  \item \textbf{Initialize.}
    \[
      a(s_1)\gets1,\quad q(v)\gets1,\quad c(s_2)\gets1,
    \qquad
      \Delta_M(s_1,v)\gets0,\quad \Delta_C(s_1,v)\gets0.
    \]
    Form initial Gibbs coupling
    \[
      \mu(s_1,v,s_2)\propto\bar\mu(s_1,v,s_2)\,a(s_1)q(v)c(s_2)
      \exp\{\Delta_M(s_1,v)(s_2-s_1)+\Delta_C(s_1,v)(L(s_2/s_1)-v^2)\}.
    \]
  \item \textbf{Outer loop:} repeat until marginal errors $\le\varepsilon_{\mathrm{marg}}$ and financial residuals $\le\varepsilon_{\mathrm{fin}}$:
    \begin{enumerate}
      \item \textbf{Sinkhorn marginal updates:} for all grid points update
        \[
          a(s_1)\leftarrow
          \frac{\mu_1(s_1)}{\sum_{v,s_2}\bar\mu(s_1,v,s_2)\,q(v)c(s_2)\,
            e^{\Delta_M(s_1,v)(s_2-s_1)+\Delta_C(s_1,v)(L(s_2/s_1)-v^2)}},
        \]
        \[
          q(v)\leftarrow
          \frac{\mu_V(v)}{\sum_{s_1,s_2}\bar\mu(s_1,v,s_2)\,a(s_1)c(s_2)\,
            e^{\Delta_M(s_1,v)(s_2-s_1)+\Delta_C(s_1,v)(L(s_2/s_1)-v^2)}},
        \]
        \[
          c(s_2)\leftarrow
          \frac{\mu_2(s_2)}{\sum_{s_1,v}\bar\mu(s_1,v,s_2)\,a(s_1)q(v)\,
            e^{\Delta_M(s_1,v)(s_2-s_1)+\Delta_C(s_1,v)(L(s_2/s_1)-v^2)}}.
        \]
        Refresh $\mu(s_1,v,s_2)$ using the Gibbs map above.
      \item \textbf{Conditional Newton/LM enforcement:} for each $(s_1,v)\in\mathcal{S}_1\times\mathcal{V}$ do
        \begin{enumerate}
          \item repeat up to inner iterations:
            \[
              w_{s_1,v}(s_2)\;=\;\frac{\mu(s_1,v,s_2)}{\sum_{u\in\mathcal{S}_2}\mu(s_1,v,u)}
            \]
            (conditional law over $s_2$),
            residuals
            \[
              r_M(s_1,v)=\sum_{s_2}w_{s_1,v}(s_2)(s_2-s_1),\qquad
              r_C(s_1,v)=\sum_{s_2}w_{s_1,v}(s_2)\bigl(L(s_2/s_1)-v^2\bigr),
            \]
            Jacobian entries
            \[
              J_{11}=\operatorname{Var}_{w}(s_2-s_1),\quad
              J_{22}=\operatorname{Var}_{w}\bigl(L(s_2/s_1)-v^2\bigr),
            \]
            \[
              J_{12}=J_{21}=\operatorname{Cov}_{w}\bigl(s_2-s_1,\;L(s_2/s_1)-v^2\bigr),
            \]
            solve damped Newton system
            \[
              \begin{pmatrix}J_{11}+\lambda & J_{12}\\[4pt] J_{21} & J_{22}+\lambda\end{pmatrix}
              \begin{pmatrix}\delta_M\\[2pt]\delta_C\end{pmatrix}
              =
              -\begin{pmatrix}r_M\\[2pt] r_C\end{pmatrix},
            \]
            update multipliers
            \[
              \Delta_M(s_1,v)\leftarrow\Delta_M(s_1,v)+\delta_M,\qquad
              \Delta_C(s_1,v)\leftarrow\Delta_C(s_1,v)+\delta_C,
            \]
            refresh $\mu(s_1,v,s_2)$.
          \item until $\sqrt{r_M(s_1,v)^2+r_C(s_1,v)^2}\le\varepsilon_{\mathrm{fin}}$ or inner cap reached.
        \end{enumerate}
    \end{enumerate}
  \item On convergence set $\mu^\star\gets\mu$.
  \item \textbf{Fisher information.} With sufficient statistics $T_i(x)$ compute
    \[
      H_{ij}=\sum_{x}\mu^\star(x)\bigl(T_i(x)-\mathbb{E}_{\mu^\star}[T_i]\bigr)
      \bigl(T_j(x)-\mathbb{E}_{\mu^\star}[T_j]\bigr),
    \quad x=(s_1,v,s_2).
    \]
  \item \textbf{Return:} $\mu^\star,H$.
\end{enumerate}
\normalsize
\subsection{POT Risk Computation: The Linear Response (LR) Approach}
\label{subsec:risk_computation}

We compute first-order (Gâteaux) price sensitivities under small marginal perturbations using the Fisher information matrix from the calibrated exponential family.
Let $\mu^\star$ denote the calibrated coupling, $H$ the Fisher matrix, and $h$ the perturbation vector~\eqref{eq:perturbation_vector}.
As shown in Section~\ref{subsec:ssr_compatibility}, the VIX marginal component $h_V$ is fully determined by the SPX shock via the SSR rule~\eqref{eq:hV}.
The constraint operator $\mathcal{A}$, the Fisher matrix $H$, and the dual structure are unchanged; the SSR rule acts only on the construction of $h$.

The financial constraints are preserved automatically: since the financial
block of $h$ is zero by construction~\eqref{eq:perturbation_vector},
the variation $\delta\mu = (\mu^\varepsilon - \mu^\star)/\varepsilon$ satisfies
$\mathcal{A}_{\mathrm{fin}}\delta\mu = 0$ as a consequence of the perturbed
problem structure~\eqref{eq:primal_eps}; no projection of $h$ is needed
or meaningful.

Let the payoff be $G:\mathcal{X}\to\mathbb{R}$ with price $\Pi(\varepsilon) = \mathbb{E}_{\mu^\varepsilon}[G]$ and baseline $\Pi(0)=\mathbb{E}_{\mu^\star}[G]$. The first-order risk is
$$
\Pi'(0) \;=\; g^\top \dot{\theta},
$$
where $\dot{\theta}$ solves the linear response system $H\,\dot{\theta} = h$, and $g$ is the covariance vector of $G$ with the sufficient statistics.

\paragraph{Inputs.}
Calibrated coupling $\mu^\star$; Fisher matrix $H$; full perturbation
vector $h=(h_1,h_V,h_2,\mathbf{0})\in\mathbb{R}^K$~\eqref{eq:perturbation_vector};
payoff $G(x)$; optional damping $\lambda\ge 0$ and solver tolerances.

\paragraph{Outputs.}
First-order risk $\Pi'(0)$; optionally the dual variation $\dot\theta^0$ (for greeks attribution).

\medskip
\noindent\textbf{Algorithm 2: Linear Response (LR) Risk Computation.}
\begin{enumerate}
\item \textbf{Solve the regularized linear response system.}
Compute the dual variation by solving
$$
(H + \lambda I)\,\dot\theta^0 \;=\; h,
$$
where $\lambda\geq 0$ is a Tikhonov regularization parameter.
By Lemmas~\ref{lem:range_gG}--\ref{lem:range_db} and Theorem~\ref{thm:risk_representation},
the price sensitivity $g^\top\dot{\theta}$ converges to $g^\top H^+h$
as $\lambda\downarrow 0$, because $g\in\mathrm{range}(H)$ annihilates the
$\lambda^{-1}$ null-space component of $(H+\lambda I)^{-1}h$.
In practice, $\lambda>0$ suppresses amplification of near-null directions
of $H$ and controls numerical stability; $\lambda=0$ recovers the exact
pseudoinverse solution when $H$ is numerically well-conditioned.
Use the same gauge as in calibration.

\item \textbf{Compute the covariance vector $g$.}
Let $\{T_i\}$ be the sufficient statistics (coordinates of the dual potentials). Compute
$$
g_i \;=\; \sum_{x \in X} \mu^\star(x)\,\Big(G(x) - \mathbb{E}_{\mu^\star}[G]\Big)\,\Big(T_i(x) - \mathbb{E}_{\mu^\star}[T_i]\Big),
\qquad i=1,\dots,\dim(\theta).
$$

\item \textbf{Evaluate first-order risk.}
$$
\Pi'(0) \;=\; g^\top \dot\theta^0.
$$
Return $\Pi'(0)$ (and $\dot\theta^0$ if needed for greeks attribution).
\end{enumerate}

\paragraph{Notes.}
\begin{itemize}
\item The covariance step can be reused across multiple payoffs once $\mu^\star$ and $\{T_i\}$ are fixed; only $g$ changes with $G$.
\end{itemize}
\section{Dimensional Reduction (DR) For POT}
\label{sec:dimension_reduction}
An alternative to computing first-order sensitivities via the Fisher information is to exploit the conditional
coupling invariance directly and re-solve a reduced entropic projection on $(S_1,V)$.  Under a conditional kernel invariance assumption, the perturbed three-dimensional problem is equivalent to a two-dimensional entropic projection for
$\gamma$ on $(S_1,V)$, so one may obtain the exact perturbed coupling in the reduced class by solving a
Sinkhorn-type projection that matches the perturbed VIX marginal implied by the SSR propagation of the SPX shock, while remaining close in entropy to the base reduced coupling.
This reduced-OT approach retains convexity and numerical stability and, unlike the Fisher linearization,
captures nonlinear effects for finite (non-infinitesimal) shocks insofar as the dimension reduction assumption remains numerically accurate.
\subsection{Base Conditional Structure}

Let $\mu^\star \in \Delta(\mathcal{X})$ denote the calibrated optimal coupling
on $\mathcal{X} = \mathcal{S}_1 \times \mathcal{V} \times \mathcal{S}_2$.
Define its $(S_1,V)$-marginal and the conditional kernel of $S_2$ given $(S_1,V)$ by
\[
\gamma^\star(s_1,v)
:=
\sum_{s_2 \in \mathcal{S}_2}
\mu^\star(s_1,v,s_2),
\qquad
\kappa^\star(s_2 \mid s_1,v)
:=
\frac{\mu^\star(s_1,v,s_2)}
{\gamma^\star(s_1,v)},
\]
so that the coupling admits the disintegration
\[
\mu^\star(s_1,v,s_2)
=
\kappa^\star(s_2 \mid s_1,v)\,
\gamma^\star(s_1,v).
\]


\subsection{Conditional Coupling Invariance Assumption}

\begin{assumption}[Conditional Coupling Invariance]
\label{ass:conditional_invariance}
Under sufficiently small perturbations of the SPX marginals,
the conditional distribution of $S_2$ given $(S_1,V)$ remains unchanged, i.e.,

\[
\kappa^\varepsilon(s_2 \mid s_1,v)
=
\kappa^\star(s_2 \mid s_1,v)
\quad
\text{for all } (s_1,v,s_2).
\]

\end{assumption}

This assumption reflects that the structural dependence
between $S_2$ and $(S_1,V)$ is stable under small marginal shocks.

\begin{remark}[$S_2$ marginal under DR: endogeneity and LR comparison]
\label{rem:S2_endogenous}
Under Assumption~\ref{ass:conditional_invariance}, the $S_2$ marginal of
the perturbed coupling is fully determined by $\kappa^\star$ and $\gamma^\varepsilon$:
\[
\mu_2^{DR}(s_2) \;=\; \sum_{s_1,v}\kappa^\star(s_2\mid s_1,v)\,\gamma^\varepsilon(s_1,v).
\]
This is an \emph{endogenous} output of the DR method, not an independently
prescribed target. If a spot or volatility bump prescribes an exogenous
$\mu_2^\varepsilon \neq \mu_2^{DR}$, then the assumption and the exogenous
target are generically incompatible: enforcing both would require perturbing
$\kappa^\star$.  The DR method therefore matches $\mu_1^\varepsilon$ and
$\mu_V^\varepsilon$ exactly, and accepts $\mu_2^{DR}$ as the best
$\kappa^\star$-consistent approximation to $\mu_2^\varepsilon$.

By contrast, the \emph{Linear Response} (LR) method enforces all three
marginal targets $(\mu_1^\varepsilon,\mu_V^\varepsilon,\mu_2^\varepsilon)$
as independent inputs to the Fisher system $H\dot\theta^0=h$.
The gap between the two treatments is $O(\varepsilon)$:
\[
\mu_2^\varepsilon(s_2) - \mu_2^{DR}(s_2)
\;=\; \varepsilon\bigl[h_2(s_2) - (L_{\kappa^\star}\dot\gamma^0)(s_2)\bigr]
     + O(\varepsilon^2),
\]
where $L_{\kappa^\star}\dot\gamma^0(s_2):=\sum_{s_1,v}\kappa^\star(s_2\mid s_1,v)\dot\gamma^0(s_1,v)$.
The compatibility condition $h_2 = L_{\kappa^\star}\dot\gamma^0$ that would
make the leading term vanish does not hold for generic market bumps.
The close LR--DR agreement across VIX futures and options in
Section~\ref{subsec:risk_bench} confirms that this gap is small in practice.
\end{remark}


\subsection{Exact Dimensional Reduction}

\begin{theorem}[Exact Reduction to Two-Dimensional Entropic Projection]
\label{thm:dim_reduction}
Assume $\mu^\star$ is strictly positive on $\mathcal{X}$
and Assumption~\ref{ass:conditional_invariance} holds.

Then the perturbed entropic projection problem

\[
\inf_{\mu \in \mathcal{P}^\varepsilon}
D(\mu \,\|\, \mu^\star)
\]

reduces exactly to the two-dimensional problem

\[
\inf_{\gamma \in \mathcal{Q}^\varepsilon}
D(\gamma \,\|\, \gamma^\star),
\]

where $\mathcal{Q}^\varepsilon$ is the set of probability measures $\gamma$
on $\mathcal{S}_1\times\mathcal{V}$ satisfying the perturbed
$(S_1,V)$-marginal constraints:
\begin{equation}
\label{eq:Q_eps}
\mathcal{Q}^\varepsilon \;:=\;
\Bigl\{\gamma\in\Delta(\mathcal{S}_1\times\mathcal{V}):\;
\textstyle\sum_v\gamma(s_1,v)=\mu_1^\varepsilon(s_1),\;
\sum_{s_1}\gamma(s_1,v)=\mu_V^\varepsilon(v)\Bigr\}.
\end{equation}
The $S_2$ marginal is \emph{endogenous}: it is determined as
$\mu_2^{DR}(s_2)=\sum_{s_1,v}\kappa^\star(s_2\mid s_1,v)\gamma^\varepsilon(s_1,v)$
rather than independently prescribed (see Remark~\ref{rem:S2_endogenous}).
The full three-dimensional coupling is reconstructed by

\[
\mu^\varepsilon(s_1,v,s_2)
=
\kappa^\star(s_2 \mid s_1,v)\,
\gamma^\varepsilon(s_1,v).
\]

\end{theorem}


\begin{proof}
Consider the restricted class of couplings that share the base conditional
kernel $\kappa^\star$:
\[
\mathcal{R}^\varepsilon
\;:=\;
\bigl\{\,\mu\in\mathcal{P}^\varepsilon :
\mu(s_1,v,s_2)=\kappa^\star(s_2\mid s_1,v)\,\gamma(s_1,v)
\text{ for some probability measure }\gamma\text{ on }\mathcal{S}_1\times\mathcal{V}
\,\bigr\}.
\]
By Assumption~\ref{ass:conditional_invariance}, the optimizer of the full
problem $\inf_{\mathcal{P}^\varepsilon}D(\mu\|\mu^\star)$ has its conditional
kernel equal to $\kappa^\star$ for all sufficiently small $\varepsilon$; hence it
belongs to $\mathcal{R}^\varepsilon$.
The full infimum therefore coincides with the restricted infimum:
\[
\inf_{\mu\in\mathcal{P}^\varepsilon}D(\mu\|\mu^\star)
\;=\;
\inf_{\mu\in\mathcal{R}^\varepsilon}D(\mu\|\mu^\star).
\]
It now suffices to evaluate the entropy for any $\mu\in\mathcal{R}^\varepsilon$.
Writing $\mu(s_1,v,s_2)=\kappa^\star(s_2\mid s_1,v)\,\gamma(s_1,v)$ and substituting
into the relative entropy:

\[
D(\mu^\varepsilon \,\|\, \mu^\star)
=
\sum_{s_1,v,s_2}
\mu^\varepsilon(s_1,v,s_2)
\ln
\frac{\mu^\varepsilon(s_1,v,s_2)}
{\mu^\star(s_1,v,s_2)}.
\]

Using the disintegration formulas for $\mu^\varepsilon$ and $\mu^\star$:

\[
=
\sum_{s_1,v,s_2}
\kappa^\star(s_2 \mid s_1,v)\,
\gamma(s_1,v)
\ln
\frac{\kappa^\star(s_2 \mid s_1,v)\,\gamma(s_1,v)}
{\kappa^\star(s_2 \mid s_1,v)\,\gamma^\star(s_1,v)}.
\]

Canceling $\kappa^\star(s_2 \mid s_1,v)$ inside the logarithm yields

\[
=
\sum_{s_1,v,s_2}
\kappa^\star(s_2 \mid s_1,v)\,
\gamma(s_1,v)
\ln
\frac{\gamma(s_1,v)}
{\gamma^\star(s_1,v)}.
\]

Since for each $(s_1,v)$,
$\sum_{s_2}\kappa^\star(s_2 \mid s_1,v)=1$,
we obtain

\[
D(\mu^\varepsilon \,\|\, \mu^\star)
=
\sum_{s_1,v}
\gamma(s_1,v)
\ln
\frac{\gamma(s_1,v)}
{\gamma^\star(s_1,v)}
=
D(\gamma \,\|\, \gamma^\star).
\]

Thus the three-dimensional projection problem
is equivalent to the two-dimensional entropic projection.

The perturbed constraints on $\mu$ reduce to the $(S_1,V)$-marginal
constraints~\eqref{eq:Q_eps} on $\gamma$; the $S_2$ marginal is
determined endogenously by $\kappa^\star\gamma^\varepsilon$.
The reconstructed coupling preserves the conditional martingale and
variance-consistency relations inherited from the base calibration,
since these depend only on the fixed conditional kernel $\kappa^\star$.
\end{proof}
\paragraph{Computational implication.}

In the base calibration, the entropic martingale optimal transport problem
must enforce the martingale constraint $\mathbb{E}[S_2|S_1,V]=S_1$,
which couples all three variables and requires a full three-dimensional
Sinkhorn calibration.
By Theorem~\ref{thm:dim_reduction}, the perturbed problem requires no such
re-imposition: operationally it amounts to running a two-marginal Sinkhorn
projection on $(S_1,V)$ \emph{without re-imposing the martingale constraint},
with $\kappa^\star$ carrying the financial constraints forward automatically.

This is the key reason why risk generation under DR is computationally
efficient: the perturbed problem no longer requires recalibration of the
full martingale optimal transport model.
In practice the perturbed projection typically converges in only 5--6
Sinkhorn iterations because the solution is close to the base coupling.

\subsection{SPX--VIX Family Risk Generation: The Dimensional Reduction (DR) Approach}

We now describe the practical algorithm used to compute SPX--VIX risk
sensitivities under SSR while preserving the structure of the calibrated
joint coupling $\mu^\star(s_1,v,s_2)$ from Section~\ref{subsec:base},
which by construction satisfies the SPX and VIX market constraints,
the martingale condition, and the SPX--VIX consistency condition.
The key point is that the perturbation is not generated by directly
changing the SPX marginal inside the transport problem. Rather, one starts
from an exogenous SPX market shock (a spot or volatility bump), propagates
it through the SSR dynamics to obtain the corresponding change in VIX option
prices, and uses these to define a new VIX marginal target.
Rather than fully recalibrating the joint coupling, which would be
computationally expensive, we exploit the dimension reduction result of
Section~\ref{sec:dimension_reduction}: the perturbation is carried out at
the level of the lower-dimensional coupling on $(S_1,V)$ while leaving the
conditional kernel of $S_2$ given $(S_1,V)$ unchanged.

More precisely, write the base calibrated coupling in disintegrated form as
\[
\mu^\star(s_1,v,s_2)
=
\gamma^\star(s_1,v)\,\kappa^\star(s_2\mid s_1,v),
\]
where $\gamma^\star(s_1,v)$ is the marginal coupling of $(S_1,V)$ and
$\kappa^\star(s_2\mid s_1,v)$ is the conditional kernel of $S_2$ given
$(S_1,V)$.
By the dimension reduction theorem, the perturbed coupling is constructed by
updating only $\gamma^\star$ while holding $\kappa^\star$ fixed:
\[
\mu_\varepsilon(s_1,v,s_2)
=
\gamma^\varepsilon(s_1,v)\,\kappa^\star(s_2\mid s_1,v).
\]
The VIX marginal therefore adjusts through the perturbation of
$\gamma^\varepsilon$, while the conditional dependence structure of $S_2$
given $(S_1,V)$ remains inherited from the base calibration.

\medskip
\noindent\textbf{Algorithm 3: SSR-Enhanced Dimensional Reduction (DR) For POT Risk Generation}

\begin{enumerate}
  \item \textbf{Inputs.}
    \begin{itemize}
      \item Base joint coupling $\mu^\star(s_1, v, s_2)$ and relevant marginals from Algorithm 1
      \item Exogenous SPX perturbation (e.g., spot bump or volatility surface shift)
      \item SSR (Skew Stickiness Ratio) parameters for VIX volatility dynamics
      \item Observed SPX and VIX market data
    \end{itemize}

  \item \textbf{Outputs.}
    \begin{itemize}
      \item Updated perturbed coupling $\mu^\varepsilon(s_1, v, s_2)$
      \item Risk sensitivities under $\mu_\varepsilon$
    \end{itemize}

  \item \textbf{Base Calibration.}
    \begin{enumerate}      
      \item Disintegrate $\mu^\star$ as
      \[
         \mu^\star(s_1, v, s_2) = \gamma^\star(s_1, v)\,\kappa^\star(s_2 \mid s_1, v),
      \]
      where $\gamma^\star$ is the $(S_1, V)$ marginal and $\kappa^\star$ is the conditional kernel.
    \end{enumerate}

  \item \textbf{Generate exogenous SPX perturbation.}
    \begin{enumerate}
      \item Apply the prescribed SPX perturbation (e.g., spot or implied volatility shift) to obtain the new SPX marginal and updated SPX implied forward variance $F_V'$.
    \end{enumerate}

  \item \textbf{Propagate VIX smile using SSR.}
    \begin{enumerate}
      \item Apply the SSR rule~\eqref{eq:VIXLinearVolShift} to obtain the perturbed VIX implied volatility surface $\sigma'_{\mathrm{VIX}}(K)$.
      \item Compute the corresponding perturbed VIX option prices from the shifted surface.
    \end{enumerate}

  \item \textbf{Construct perturbed VIX marginal.}
    \begin{enumerate}
      \item Infer the new VIX marginal distribution so that, under the VIX variable, the model reproduces the SSR-propagated VIX option prices.
    \end{enumerate}

  \item \textbf{Dimension-reduced entropic transport update.}
    \begin{enumerate}
      \item Holding $\kappa^\star(s_2\mid s_1, v)$ fixed, solve the
      two-dimensional entropic projection~\eqref{eq:Q_eps}: find
      $\gamma^\varepsilon\in\mathcal{Q}^\varepsilon$ closest in
      $D(\cdot\|\gamma^\star)$, where $\mathcal{Q}^\varepsilon$ enforces
      the perturbed \emph{SPX $T_1$ marginal} $\mu_1^\varepsilon$ (from the
      SPX shock) and the perturbed \emph{VIX marginal} $\mu_V^\varepsilon$
      (from the SSR propagation).
      The $S_2$ marginal is \emph{not} independently constrained; it is
      determined endogenously as
      $\mu_2^{DR}(s_2)=\sum_{s_1,v}\kappa^\star(s_2\mid s_1,v)\gamma^\varepsilon(s_1,v)$
      (see Remark~\ref{rem:S2_endogenous}).
      \item Reconstruct the full perturbed joint coupling as
      \[
        \mu_\varepsilon(s_1, v, s_2)
        = \gamma^\varepsilon(s_1, v)\,\kappa^\star(s_2\mid s_1, v).
      \]
    \end{enumerate}

  \item \textbf{Risk extraction.}
    \begin{enumerate}
      \item For any payoff function $G$, compute model prices under $\mu^\star$ and $\mu_\varepsilon$, and report the sensitivity as
      \[
        \text{Greek}
        \approx
        \frac{P(\mu_\varepsilon)-P(\mu^\star)}{\varepsilon},
      \]
      where $P(\mu)$ denotes the portfolio valuation under $\mu$.
    \end{enumerate}

\end{enumerate}
\medskip
This procedure avoids a full recalibration of the original SPX--VIX
martingale optimal transport problem. The perturbation is carried only
by the reduced coupling in $(S_1,V)$, while the conditional kernel of
$S_2$ given $(S_1,V)$ remains unchanged. In this way, the endogenous
adjustment of the VIX marginal is captured through the perturbed reduced
coupling, yielding a fast and structurally consistent risk-generation
algorithm.
\section{Experiments: SPX-VIX Risk Generation And Hedging Backtest}
\label{sec:experiments}
\subsection{SPX–VIX Basis In Market Data}

In theory the VIX future level should be consistent with the forward variance
implied by the SPX option surface through the well-known replication formula.  This is precisely the consistency condition in \eqref{eq:consistency_condition}.  However, empirical market data shows that this relation does not hold exactly.
In practice a persistent basis exists between the SPX implied forward variance
and the traded VIX futures.

Figure~\ref{fig:spx_vix_basis} shows the time series of the SPX–VIX basis for
the 1-month tenor over a two-year window; this basis is well documented in
practice and is typically attributed to market segmentation, liquidity effects,
and supply–demand imbalances in the VIX futures market.
Our calibration framework therefore allows a basis adjustment when linking SPX
forward variance to the VIX future level.

\paragraph{Formal treatment of the basis.}
Define the pointwise SPX--VIX basis at each node $(s_1,v)$ by
\begin{equation}
\label{eq:eta_basis}
\eta(s_1,v)
\;:=\;
v^2 \;-\; \mathbb{E}_{\mu^\star}\!\bigl[L(S_2/S_1)\mid S_1=s_1,\,V=v\bigr],
\end{equation}
so that $\eta$ measures the pointwise gap between the VIX-squared target and
the SPX-implied forward variance realized under the calibrated coupling.
The consistency relation~\eqref{eq:consistency_condition} corresponds to
$\eta\equiv 0$; empirically, as documented in Figure~\ref{fig:spx_vix_basis},
$\eta$ is persistently nonzero due to convexity adjustments, discrete
monitoring, and VIX settlement conventions.
We therefore replace~\eqref{eq:consistency_condition} by the relaxed
node-wise target
\begin{equation}
\label{eq:consistency_relaxed}
(\mathcal{A}_{\mathrm{fin}}\mu)^C_{ij}
\;=\;
\sum_{s_2}\mu(s_1^i,v^j,s_2)\,
\bigl(L(s_2/s_1^i)-(v^j)^2+\eta(s_1^i,v^j)\bigr)
\;=\;0,
\end{equation}
i.e.\ $\mathbb{E}[L(S_2/S_1)\mid S_1,V]=V^2-\eta(S_1,V)$ at each grid node,
with $\eta$ entering the Gibbs exponent additively through $\Delta_C(s_1,v)$
(Algorithm~1).
The basis $\eta$ is calibrated once from market data (specifically, from the
gap between the SPX-implied variance strip and the quoted VIX future at each
conditioning node $(s_1^i,v^j)$), and is treated as a fixed structural offset
that does not vary with the market perturbation.
It is enforced softly via the financial tolerance $\varepsilon_{\mathrm{fin}}$
in Algorithm~1: residuals of size $O(\varepsilon_{\mathrm{fin}})$ around the
$\eta$-shifted target are absorbed by the inner Newton stopping criterion.

Under the spot and volatility perturbations of
Section~\ref{sec:financial_perturbations}, $\eta$ is held fixed: bumps act
only on the marginal targets $(\mu_1^\varepsilon,\mu_V^\varepsilon,
\mu_2^\varepsilon)$, so $db^\varepsilon_{\mathrm{fin}}/d\varepsilon=0$ and
the financial slots of the perturbation vector $h$~\eqref{eq:perturbation_vector}
remain zero.
This is appropriate as long as the basis is driven by structural
microstructure effects (liquidity, settlement methodology) that do not
co-move with small SPX surface perturbations.

\begin{figure}[h]
\centering
\includegraphics[width=0.6\textwidth]{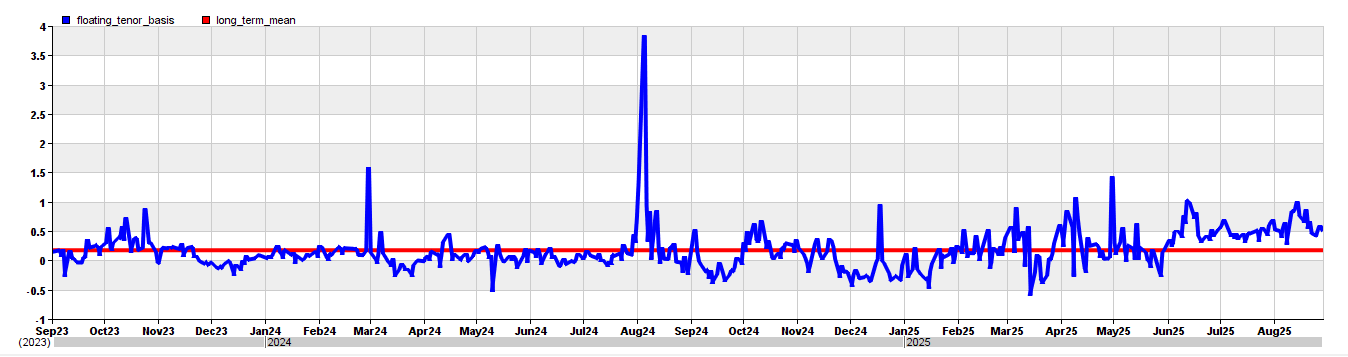}
\caption{SPX–VIX basis time series for the 1-month tenor over a two-year window.}
\label{fig:spx_vix_basis}
\end{figure}

\subsection{Base Calibration And Fit Quality}
\label{subsec:fit_quality}

Despite the SPX–VIX basis noted above, the calibrated optimal transport model satisfies the martingality condition and approximate consistency constraints, as we verify below.
We first assess the quality of the base calibration used throughout the experiments. The joint SPX--VIX coupling $\mu^\star$ is obtained via entropic projection with Sinkhorn scaling on the discrete grids. After calibration, we compute the sufficient statistics and the Fisher information matrix $H$ for subsequent risk analysis. Existence, uniqueness, and the Gibbs form of the optimizer ensure a consistent fit to the prescribed marginals and targets on the grids.

To visualize fit quality, we plot observed market smiles against model-implied values from the calibrated coupling for a representative maturity (e.g., March 18, 2026, two weeks to expiry), consistent with our option risk comparisons. Additionally, we generate the martingality plot and consistency plot.


\begin{figure}[H]
\centering
\begin{subfigure}[t]{0.48\textwidth}
  \centering
  \includegraphics[height=0.45\textwidth]{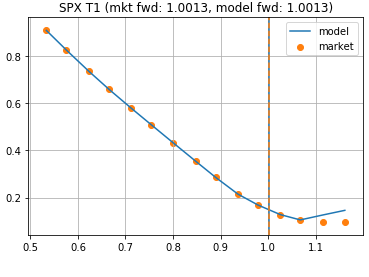}
  \caption{SPX $T_1$ smile fit}
  \label{fig:spx_t1}
\end{subfigure}
\hfill
\begin{subfigure}[t]{0.49\textwidth}
  \centering
  \includegraphics[height=0.45\textwidth]{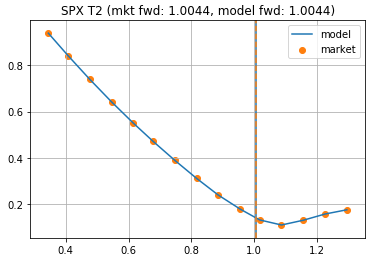}
  \caption{SPX $T_2$ smile fit}
  \label{fig:spx_t2}
\end{subfigure}
\caption{SPX smile fit quality (two expiries)}
\label{fig:spx_fit_pair}
\end{figure}

\begin{figure}[H]
\centering
\includegraphics[width=0.6\linewidth]{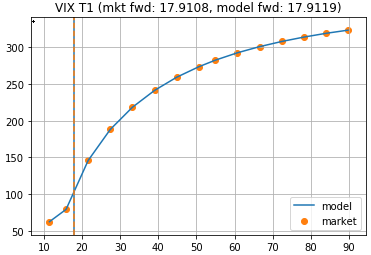}
\caption{VIX smile fit: observed vs model-implied}
\label{fig:vix_fit}
\end{figure}

\begin{figure}[H]
\centering
\begin{subfigure}[t]{0.48\textwidth}
  \centering
  \includegraphics[height=0.45\textwidth]{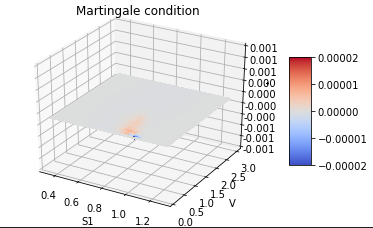}
  \caption{Martingality check}
  \label{fig:martingality}
\end{subfigure}
\hfill
\begin{subfigure}[t]{0.49\textwidth}
  \centering
  \includegraphics[height=0.45\textwidth]{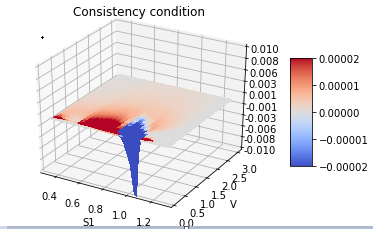}
  \caption{Consistency condition}
  \label{fig:consistency}
\end{subfigure}
\caption{SPX--VIX calibration theoretical conditions}
\label{fig:martingality_consistency}
\end{figure}

\FloatBarrier

\subsection{Risk Computation Benchmark: Recalibration vs Perturbation}
\label{subsec:risk_bench}

This section evaluates the accuracy and computational efficiency of the
perturbation framework developed in the preceding sections by comparing two approaches
for computing risk sensitivities under SPX market perturbations.

\begin{enumerate}

\item \textbf{Full recalibration (benchmark).}
After applying a perturbation to the SPX market, the entire
SPX–VIX martingale optimal transport calibration problem
is recomputed using Algorithm~1.
The resulting joint distribution
$\mu^{\varepsilon}_{\mathrm{recal}}$
serves as the benchmark distribution for computing option prices
and sensitivities.

\item \textbf{Perturbation (LR).}
Starting from the calibrated base coupling $\mu^\star$,
we compute the perturbed distribution using the
linear response system derived in Sections~\ref{subsec:risk_rep}--\ref{subsec:second_order}.
This method uses the Fisher information matrix of the calibrated
exponential family to approximate the perturbed coupling
$\mu^{\varepsilon}_{\mathrm{pert}}$
without solving the full calibration problem again.

\end{enumerate}
The goal is twofold: to verify that the perturbation-based sensitivities
closely match those from full recalibration, and to demonstrate the
significant computational speedup the LR method achieves over repeated
full calibration (Table~\ref{tab:perf}).

\paragraph{SPX perturbations.}

In all experiments the perturbation is applied on the SPX side,
either as a spot shift or as a parallel shift of the SPX implied
volatility surface. These perturbations are mapped to corresponding
changes in the forward variance, which acts as the key control
variable in the SPX–VIX coupling.

The perturbations are chosen to remain within a regime where the
linear-response approximation is expected to be accurate while still
representing realistic market shocks.

\paragraph{Implementation details.}

Several groups of parameters control the perturbation experiments:

\begin{itemize}

\item \textbf{Base OT object.}
The initial martingale optimal transport calibration provides the
reference coupling $\mu^\star$ and the Fisher information matrix
used in the perturbation calculations.

\item \textbf{Bumped SPX information.}
The perturbed SPX marginal distributions include the shifted spot
and the modified implied volatility surface at the relevant maturities.
Throughout the experiments we assume a sticky-strike behavior for the
SPX volatility surface under spot perturbations.

\item \textbf{Perturbation controls.}
Parameters defining the magnitude and structure of the volatility
perturbations, including lower and upper cutoffs for invariant
volatility regions.

\item \textbf{VIX volatility shape controls.}
Parameters governing the Skew Stickiness Ratio (SSR), skewness,
convexity, and the treatment of at-the-money and out-of-the-money
VIX option strikes.

\item \textbf{Basis and numerical controls.}
Optional parameters allowing forward basis adjustments,
regularization parameters, and marking conventions for volatility,
skew, SSR, convexity, and VIX marginal constraints.

\end{itemize}

\subsection{VIX Option Risk And Cross-Greeks}

Using the experimental setup of Section~\ref{subsec:risk_bench}, we compare
SPX risk sensitivities of VIX derivatives under full recalibration and
the LR perturbation method.

\paragraph{VIX future cross-greeks.}

We begin by comparing the SPX cross-greeks of the VIX future contract.
The sensitivities are computed with respect to SPX spot and SPX implied
volatility perturbations. Table~\ref{tab:vix_future_risk} reports the SPX delta and SPX vega
of the VIX future obtained from the recalibration benchmark and from
the perturbation method.

\begin{center}
\begin{table}[!htbp]
\centering
\small
\setlength{\tabcolsep}{6pt}
\renewcommand{\arraystretch}{0.9}
\caption{VIX Future's SPX Greeks: LR-POT vs Recalibration}
\label{tab:vix_future_risk}
\begin{tabular}{l r r r r}
\hline
 & \textbf{LR-POT. SPX Delta} & \textbf{Recalib. SPX Delta} & \textbf{LR-POT. SPX Vega} & \textbf{Recalib. SPX Vega} \\
\hline
VIX Future & -8.88339 & -8.99380 & 1043.42832 & 1071.35419 \\
\hline
\end{tabular}
\end{table}
\end{center}

The results show that the perturbation-based sensitivities closely
match those obtained from the full recalibration procedure. The
differences remain small relative to the magnitude of the sensitivities,
confirming that the linear response (LR) system provides an accurate local
approximation of the recalibrated optimal transport model.

\paragraph{VIX option SPX delta.}

We next compare SPX delta sensitivities for a strip of VIX call options
spanning a wide range of strikes. The options correspond to the same
expiry used in the calibration experiment and cover both out-of-the-money
and near-the-money regions of the VIX smile.


\begin{center}
\begin{table}[!htbp]
\centering
\small
\setlength{\tabcolsep}{3.5pt}
\renewcommand{\arraystretch}{0.9}
\caption{VIX Options SPX Delta Comparison (LR-POT vs Recalibration). March 18, 2026, 2w to expiry.}
\label{tab:vix_spx_delta_compact}
\begin{tabular}{S[table-format=2.1] S[table-format=1.3] S[table-format=1.3]}
\toprule
{Strike} & {Pert. Delta} & {Recalib. Delta} \\
\midrule
16.9 & -0.786 & -0.809 \\
17.6 & -0.750 & -0.763 \\
18.2 & -0.714 & -0.719 \\
18.7 & -0.680 & -0.674 \\
19.2 & -0.641 & -0.629 \\
19.7 & -0.601 & -0.585 \\
20.3 & -0.559 & -0.540 \\
20.9 & -0.516 & -0.495 \\
21.6 & -0.469 & -0.450 \\
22.5 & -0.420 & -0.405 \\
23.5 & -0.369 & -0.360 \\
24.7 & -0.316 & -0.315 \\
26.3 & -0.266 & -0.270 \\
28.5 & -0.221 & -0.225 \\
31.4 & -0.178 & -0.179 \\
35.8 & -0.133 & -0.135 \\
44.0 & -0.089 & -0.090 \\
\bottomrule
\end{tabular}
\end{table}
\begin{figure}[H]
\centering
\includegraphics[width=\linewidth]{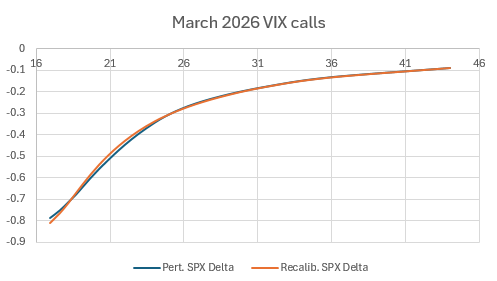} 
\caption{SPX Delta}\label{fig:SPX Delta risk}
\end{figure}
\end{center}

Figure~\ref{fig:SPX Delta risk} visualizes the same comparison across strikes.

The perturbation-based deltas closely track the sensitivities obtained
from the full recalibration. Small discrepancies appear primarily in
the wings of the VIX smile, where nonlinear effects become more
pronounced. However, even in these regions the overall shape and
magnitude of the sensitivities remain consistent with the recalibration
benchmark.

\paragraph{VIX option SPX vega.}

We perform the same comparison for SPX vega sensitivities of the VIX
options. Table~\ref{tab:vix_spx_vega_compact} reports the SPX vega obtained from both approaches.

\begin{center}
\begin{table}[!htbp]
\centering
\small
\setlength{\tabcolsep}{3.5pt}
\renewcommand{\arraystretch}{0.9}
\caption{VIX Options SPX Vega Comparison (LR-POT vs Recalibration). March 18, 2026, 2w to expiry.}
\label{tab:vix_spx_vega_compact}
\begin{tabular}{S[table-format=2.1] S[table-format=2.3] S[table-format=2.3]}
\toprule
{Strike} & {Pert. Vega} & {Recalib. Vega} \\
\midrule
16.9 & 94.452 & 96.675 \\
17.6 & 90.136 & 91.291 \\
18.2 & 85.485 & 86.048 \\
18.7 & 81.221 & 80.708 \\
19.2 & 76.337 & 75.456 \\
19.7 & 71.155 & 70.103 \\
20.3 & 65.942 & 64.734 \\
20.9 & 60.749 & 59.364 \\
21.6 & 55.294 & 53.985 \\
22.5 & 50.338 & 48.601 \\
23.5 & 45.134 & 43.213 \\
24.7 & 39.681 & 37.821 \\
26.3 & 34.065 & 32.419 \\
28.5 & 28.130 & 27.014 \\
31.4 & 21.819 & 21.568 \\
35.8 & 15.672 & 16.206 \\
44.0 & 10.432 & 10.800 \\
\bottomrule
\end{tabular}
\end{table}
\begin{figure}[H]
\centering
\includegraphics[width=\linewidth]{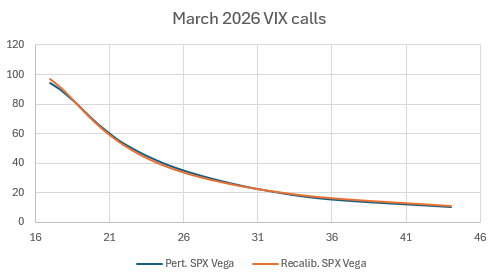} 
\caption{SPX Vega}\label{fig:SPX Vega risk}
\end{figure}
\end{center}
As with the delta comparison, the perturbation method reproduces the
recalibrated sensitivities with high accuracy. The agreement confirms
that the Fisher-information linear response captures the dominant
first-order effects of SPX volatility perturbations on VIX option
prices.

\paragraph{Performance comparison.}

While the recalibration method provides the benchmark sensitivities,
it requires solving the full SPX--VIX optimal transport calibration
problem after each perturbation. This involves repeated Sinkhorn
iterations together with enforcement of the martingale and
variance-consistency constraints, making the computation relatively
expensive.
\begin{table}[ht]
\centering
\caption{Risk Performance Comparison}\label{tab:perf}
\begin{tabular}{l l l l}
\toprule
Base calibration & Recalib based risk calculation & DR-POT method & LR-POT method \\
\midrule
348.26\,s & 370.13\,s & 18.45\,s & 5.73\,s \\
\bottomrule
\end{tabular}
\end{table}

Table \ref{tab:perf} compares the runtime required to compute sensitivities under
the three approaches. The results
show that the LR method and the DR method both achieve orders-of-magnitude speedups (${\sim}60\times$ and ${\sim}20\times$ respectively; see Table~\ref{tab:perf}) while maintaining accuracy comparable to the recalibration benchmark.

This efficiency gain is the key practical advantage of the perturbation
framework: risk sensitivities can be generated quickly without
re-running the full martingale optimal transport calibration.  

Lastly, we note that the risk numbers generated by the DR method
(Section~\ref{sec:dimension_reduction}) are extremely close to those
produced by the LR method.
This agreement is non-trivial: as discussed in
Remark~\ref{rem:S2_endogenous}, the two methods treat the $S_2$ marginal
differently: LR enforces the prescribed $\mu_2^\varepsilon$ directly,
while DR treats it as endogenous under the fixed conditional kernel
$\kappa^\star$.
The close agreement provides empirical validation that
Assumption~\ref{ass:conditional_invariance} is a good approximation in
the market environments tested, i.e.\ that $\mu_2^{DR}\approx\mu_2^\varepsilon$
to within risk-relevant accuracy.
\subsection{Backtest: Hedging Efficiency Of Optimal Transport Method}

We now evaluate whether the SPX sensitivities produced by our model independent risk generation methods lead to more effective hedging
than those generated by a benchmark industry standard model, in this case a stochastic local volatility model.  In the backtest we perform below, we choose the dimensional reduction (DR) method, Algorithm~3.   

The experiment consists of two parts: a first hedging backtest in which portfolios of
VIX options are hedged using VIX futures; and a second hedging backtest in which the same option portfolio is hedged using SPX futures and SPX vanillas. The sizing of the VIX futures
is determined by matching SPX Vega computed under
either the optimal transport method or the stochastic local volatility (SLV) benchmark.  The sizing of the SPX futures and SPX vanillas is similarly determined by matching SPX delta and SPX vega between the hedging instruments and the VIX option portfolio for each of the methods.  Because the two methodologies trade comparable notionals, we omit transaction costs.

\paragraph{Backtest period}

The backtest runs daily from January 2024 to February 2026.
VIX smile dynamics follow the SSR rule~\eqref{eq:VIXLinearVolShift}
with $\beta_V = 1.2$.

\paragraph{Synthetic portfolio generation}

To test the robustness of the hedging performance we generate
$50$ randomized VIX option portfolios.

For each trading day $t$ the portfolios are constructed as follows:

\begin{enumerate}

\item All listed VIX expiries available on day $t$ are included.

\item For each expiry we construct a strike grid using call option deltas

\[
\Delta \in \{10,15,\ldots,90\},
\]

resulting in $17$ strikes per maturity.

\item Options with $\Delta < 50$ are taken as puts while options
with $\Delta \ge 50$ are taken as calls.

\item Each option $i$ is assigned a random portfolio weight

\[
w_{i,t} \sim \text{Uniform}(-1,1).
\]

\end{enumerate}

The resulting portfolio value is

\[
P_t = \sum_i w_{i,t} V_{i,t}.
\]

This procedure produces diversified portfolios spanning a wide range
of smile exposures.

\paragraph{Hedging methodology}

The portfolios are hedged using VIX futures whose expiries match those
of the VIX options. The hedge sizes are determined by matching SPX
Vega per expiry.

For a given model $M \in \{\text{SLV},\text{POT}\}$ we compute

\[
G_t^{M} = \frac{\partial P_t}{\partial \sigma_{SPX}},
\]

the SPX sensitivity of the option portfolio, and

\[
g_{j,t}^{M} = \frac{\partial F_{j,t}}{\partial \sigma_{SPX}},
\]

the SPX sensitivity of each VIX future $F_{j,t}$.

The hedge sizes $\alpha_{j,t}^{M}$ are chosen so that

\[
G_t^{M}
+
\sum_j \alpha_{j,t}^{M} g_{j,t}^{M}
=
0.
\]

\paragraph{Hedged P\&L}

The daily hedged P\&L for model $M$ is

\[
P\&L_t^{M}
=
\Delta P_t
+
\sum_j \alpha_{j,t}^{M} \Delta F_{j,t}.
\]

The effectiveness of the hedge is evaluated using the standard
deviation of the hedged P\&L.

\paragraph{Cross-sectional comparison across portfolios}

For each of the $50$ randomly generated portfolios we compute the
standard deviation of the hedged P\&L under both hedging strategies.

Figure~\ref{fig:pnl_std_comparison} plots the difference

\[
\sigma_{POT} - \sigma_{SLV}
\]

for each portfolio, where $\sigma_{POT}$ and $\sigma_{SLV}$ denote the
standard deviation of hedged P\&L under the POT and SLV hedges,
respectively.

\begin{figure}[htbp]
\centering
\includegraphics[width=0.85\textwidth]{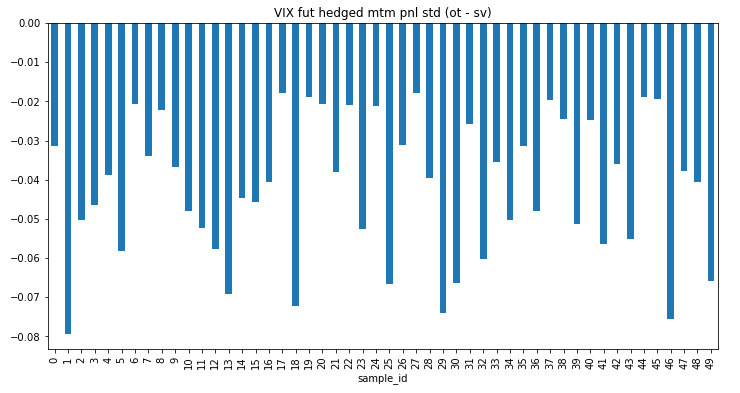}
\caption{Difference in hedged P\&L standard deviation between the
POT hedge and the SLV hedge across $50$ randomized VIX option portfolios using VIX futures.
Each bar corresponds to one portfolio.
Negative values indicate that the POT hedge achieves lower hedging
variance than the SLV hedge.}
\label{fig:pnl_std_comparison}
\end{figure}

\paragraph{Time-series hedge stability}

To illustrate the time-series behavior of the hedging error,
we select one representative portfolio from the set of
$50$ portfolios and compute the rolling $20$-day standard
deviation of hedged P\&L.

\[
\text{RollStdev}_t^{M}(20)
=
\sqrt{
\frac{1}{19}
\sum_{u=t-19}^{t}
\left(
P\&L_u^{M}
-
\overline{P\&L}_{t,20}^{M}
\right)^2
}.
\]

\begin{figure}[htbp]
\centering
\includegraphics[width=0.85\textwidth]{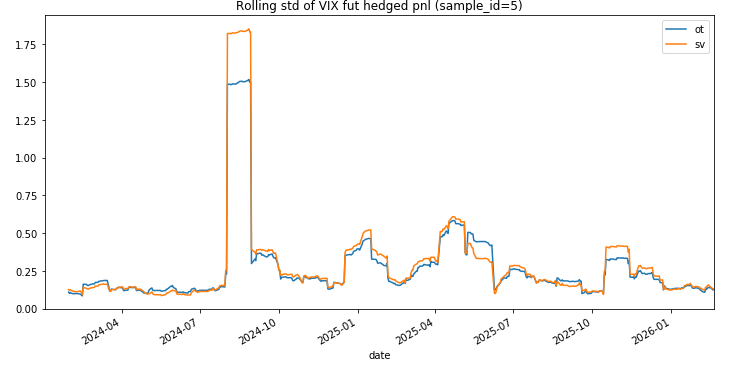}
\caption{20-day rolling standard deviation of VIX future hedged P\&L for a
representative portfolio. The POT hedge produces lower hedging
variance during volatile periods while remaining comparable to
the SLV hedge during calm market regimes.}
\label{fig:rolling_stdev}
\end{figure}

\begin{figure}[htbp]
\centering
\includegraphics[width=0.85\textwidth]{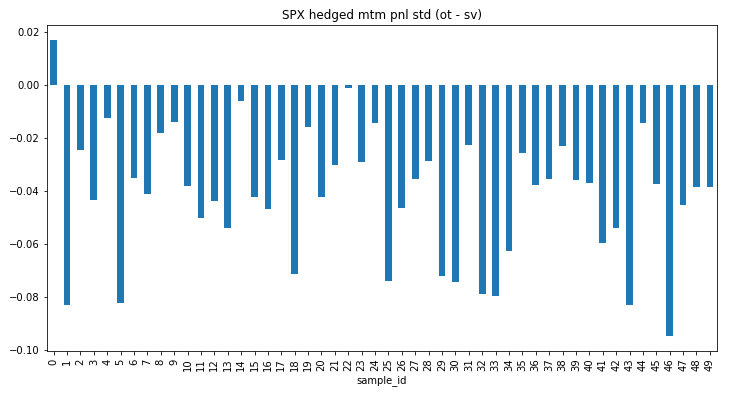}
\caption{Difference in hedged P\&L standard deviation between the
POT hedge and the SLV hedge across $50$ randomized VIX option portfolios using SPX futures and SPX vanillas.
Each bar corresponds to one portfolio.
Negative values indicate that the POT hedge achieves lower hedging
variance than the SLV hedge.}
\label{fig:pnl_std_comparison_spx}
\end{figure}

\begin{figure}[htbp]
\centering
\includegraphics[width=0.85\textwidth]{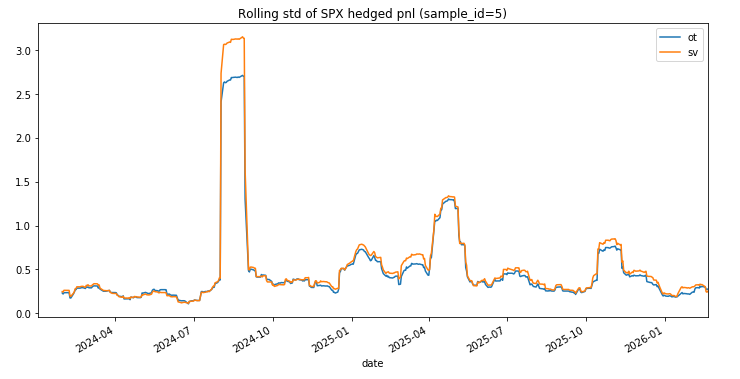}
\caption{20-day rolling standard deviation of SPX future and SPX vanillas hedged P\&L for a
representative portfolio. The POT hedge produces lower hedging
variance during volatile periods while remaining comparable to
the SLV hedge during calm market regimes.}
\label{fig:rolling_stdev_spx}
\end{figure}

\paragraph{Summary}

The cross-sectional experiment in Figure~\ref{fig:pnl_std_comparison}
shows that the POT VIX future hedge reduces PnL variance for all tested
portfolios. In Figure~\ref{fig:pnl_std_comparison_spx}, all but one of the 50 VIX option portfolios have smaller PnL variance for the POT method when hedged to SPX futures and SPX vanillas.  The time-series analysis in both Figure~\ref{fig:rolling_stdev} and Figure~\ref{fig:rolling_stdev_spx}
further demonstrates that the improvement is most pronounced
during periods of elevated market volatility.

Together these results provide empirical evidence that the
perturbed optimal transport framework produces more accurate
SPX--VIX risk sensitivities than a benchmark stochastic local volatility model.

\section{Payoff-Agnostic Risk: A VIX-Barrier Forward-Starting Option}
\label{sec:exotic}

The risk representation of Theorem~\ref{thm:risk_representation} holds for any
payoff $G:\mathcal{X}\to\mathbb{R}$, not only VIX vanillas.
To illustrate this generality we price and risk a forward-starting SPX call
whose knock-in is contingent on the VIX level at $T_1$: the return leg
$\max(S_2/S_1-1,0)$ carries no explicit VIX dependence, yet the VIX-conditional
barrier couples the payoff directly to the VIX marginal.
This exotic exercises the full three-way structure of $(S_1,V,S_2)$ and shows
that a single Fisher linear solve suffices to risk it.

\paragraph{Setup.}
We value on 2026-03-02 with $T_1=$~2026-03-18 and $T_2=$~2026-04-17.
The SPX spot is $S_0=6881.62$, with forwards $F_1=6888.80$ at $T_1$ and
$F_2=6908.70$ at $T_2$; the VIX future at $T_1$ is $F_V=20.93$.
We compare two payoffs on $\mathcal{X}=\mathcal{S}_1\times\mathcal{V}\times\mathcal{S}_2$,
\begin{align}
\label{eq:Gvan}
G_{\mathrm{van}}(S_1,V,S_2) &\;=\; \max\!\bigl(S_2/S_1-1,\,0\bigr),\\
\label{eq:Gbar}
G_{\mathrm{bar}}(S_1,V,S_2) &\;=\; \max\!\bigl(S_2/S_1-1,\,0\bigr)\cdot
\mathbf{1}_{\{V>25\}},
\end{align}
an at-the-money forward-starting SPX call~\eqref{eq:Gvan} and the same call
knocked in only when VIX exceeds $25$ at $T_1$~\eqref{eq:Gbar}.

\paragraph{Pricing.}
Table~\ref{tab:exotic_pricing} reports the two prices under the calibrated
coupling $\mu^\star$.
The barrier removes $89.20\%$ of probability mass (the scenarios with
$V\le 25$), reducing the price by $73.4\%$.
At current levels ($F_V=20.93<25$) the knock-in sits above the VIX future in the
upper tail, so the activation probability is only $10.80\%$.

\begin{table}[ht]
\centering
\small
\caption{Pricing of the forward-starting SPX call and its VIX-barrier variant
under $\mu^\star$.}
\label{tab:exotic_pricing}
\begin{tabular}{l r r}
\toprule
Payoff & Price (PV) & Barrier activation \\
\midrule
Vanilla forward starting        & 0.019205 & 100\% \\
VIX $>25$ barrier       & 0.005117 & 10.80\% \\
\bottomrule
\end{tabular}
\end{table}

\paragraph{Risk via the Fisher linear solve.}
Both payoffs are risked from the same dual variation
$\dot\theta^0=H^+h$, computed once with the variance-anchored $h_V$ and
SSR propagation ($\beta_V=1.2$) of Section~\ref{subsec:ssr_compatibility}.
The payoff-specific sensitivity is $\Pi'(0)=g_G^\top\dot\theta^0$ with
$g_G=\mathrm{Cov}_{\mu^\star}(G,T)$.
Table~\ref{tab:exotic_greeks} reports the SPX delta $\Delta$ and SPX vega
$\mathsf{V}_{\mathrm{SPX}}$ of the two payoffs.

\begin{table}[ht]
\centering
\small
\caption{SPX Greeks of the forward-starting call and its VIX-barrier variant,
from the shared Fisher linear solve.}
\label{tab:exotic_greeks}
\begin{tabular}{l r r}
\toprule
Metric & Vanilla forward starting & VIX $>25$ barrier \\
\midrule
Price (PV)                          & 0.019205              & 0.005117 \\
SPX $\Delta$ (per \$1 spot)         & $-9.25\times 10^{-6}$ & $-1.32\times 10^{-5}$ \\
SPX $\Delta$ (\% of price)          & $-0.05\%$             & $-0.26\%$ \\
SPX $\mathcal{V}$ (per vol pt)      & 0.001114              & 0.001511 \\
SPX $\mathcal{V}$ (\% of price)     & $5.80\%$              & $29.52\%$ \\
\bottomrule
\end{tabular}
\end{table}

The barrier amplifies the delta: because $\mathbf{1}_{\{V>25\}}$ is directly
sensitive to the VIX distribution, a spot increase raises $F_V$ through the
forward variance and shifts $P(V>25)$.
With the barrier above the VIX future ($F_V=20.93<25$) the knock-in activates
only $10.80\%$ of the time, so even modest VIX shifts driven by the spot bump
have a large relative impact on barrier activation, giving
$|\Delta_{\mathrm{bar}}|>|\Delta_{\mathrm{van}}|$.
The vegas differ materially: $\mathcal{V}_{\mathrm{bar}}=0.001511$ versus
$\mathcal{V}_{\mathrm{van}}=0.001114$ in absolute terms, and $29.52\%$ versus
$5.80\%$ of price respectively.
A parallel vol shift widens the VIX distribution, pushing substantial mass
across the $V=25$ gate; at $10.80\%$ activation this additional probability is
large relative to the option price, producing the high relative vega.

\paragraph{Scale invariance and variance consistency.}
The forward return $\max(S_2/S_1-1,0)$ is algebraically scale-invariant:
scaling both $S_1$ and $S_2$ by a common factor leaves it unchanged, so a
direct bump-and-revalue at a fixed $\mu^\star$ yields exactly zero delta.
The linear-response delta is nonetheless non-zero ($-9.25\times 10^{-6}$)
because it captures the market recalibration response rather than a
frozen-measure revaluation.
A spot bump changes the SPX marginals $\mu_1,\mu_2$, hence the forward variance
$\sigma^2_{\mathrm{fwd}}$, which propagates to $F_V$ and then to $\mu_V$ through
the variance anchoring~\eqref{eq:deltaFV}.
The recalibrated $\mu^\varepsilon$ reshapes the conditional distribution
$P^\varepsilon(S_2\mid S_1,V)$ to maintain variance consistency, so
$\mathbb{E}_{\mu^\varepsilon}[G_{\mathrm{van}}]\neq\mathbb{E}_{\mu^\star}[G_{\mathrm{van}}]$
even though $G_{\mathrm{van}}$ is scale-invariant; the resulting non-zero delta therefore reflects the constrained response of the calibrated joint coupling, rather than a direct spot dependence of the payoff.

\paragraph{Payoff agnostic.}
The dual variation $\dot\theta^0=H^+h$ is computed once per market
scenario and shared across all payoffs.
The per-payoff cost is the covariance projection $g_G^\top\dot\theta^0$, which
is linear in the grid dimension and requires no additional solver call.
The framework therefore applies to any structured product on $(S_1,V,S_2)$---
cliquets, corridor variance swaps, forward volatility agreements, or exotic VIX
triggers---by computing $g_G=\mathrm{Cov}_{\mu^\star}(G,T)$ and projecting.
This is the sense in which the machinery is payoff-agnostic and extends well
beyond VIX vanilla options.

\section{Conclusion}

This paper develops a model-independent framework for SPX--VIX risk
generation based on entropic martingale optimal transport.
Starting from the joint calibration methodology of Guyon,
we show that the calibrated Gibbs coupling admits a natural perturbation
theory: admissible market shocks propagate through the Fisher information
matrix of the calibrated exponential family, yielding explicit
linear-response formulas for risk sensitivities.

To incorporate realistic VIX smile dynamics, we introduce a linearized
Skew Stickiness Ratio formulation and use it to determine the VIX marginal
target change $h_V$ in the perturbation vector $h$.
This approach allows SPX perturbations to propagate consistently to
VIX implied volatility while maintaining the convex structure of the
entropic projection problem.

We further show that the perturbed transport problem admits a structural
dimensional reduction under a conditional coupling invariance assumption.
In this regime the three-dimensional transport problem collapses to a
two-dimensional projection on $(S_1,V)$ while preserving the conditional
dependence structure inherited from the base calibration.
This explains why risk sensitivities can be generated efficiently without
re-solving the full martingale optimal transport calibration.

Two sets of numerical experiments support the theoretical framework.
First, we compare perturbation-based risk sensitivities with those
obtained from full recalibration of the SPX--VIX transport model.
Across VIX futures and VIX option cross-greeks, the perturbation
method produces sensitivities that are very close to the recalibration
benchmark while achieving significant computational speedups.
Second, we conduct hedging backtests on randomized VIX option portfolios.
Using SPX sensitivities generated by the dimension-reduced transport
method, the resulting hedges consistently outperform those based on
a stochastic local volatility (SLV) benchmark in terms of hedged P\&L variance.

A further consequence of the Fisher-based perturbation structure is that
risk generation is \emph{payoff-agnostic}: the linear solve
$H\dot\theta^0=h$ depends only on the market shock, not on the payoff.
Any derivative contributes only through its covariance vector
$g_G=\mathrm{Cov}_{\mu^\star}(G,T)$, computed once at calibration, so the
per-Greek cost for each additional instrument reduces to a single projection
$g_G^\top\dot\theta^0$.
The VIX-barrier forward-starting SPX call example makes this concrete: the
barrier and vanilla positions share exactly the same Fisher solve, with the
barrier's non-trivial cross-delta arising from variance-consistency
recalibration rather than any parametric assumption on joint SPX--VIX
dynamics.

Overall, the results show that entropic martingale optimal transport
provides more than a calibration tool.
Combined with perturbation theory and dimensional reduction,
it yields a practical framework for SPX--VIX risk generation that is
financially consistent, computationally efficient, and effective in
hedging applications.

\section*{Disclaimer}
This paper was prepared for informational purposes in part by the Quantitative Trading \& Research Group of JPMorganChase \& Co. This paper is not a product of the Research Department of JPMorganChase \& Co. or its affiliates. Neither JPMorganChase \& Co. nor any of its affiliates makes any explicit or implied representation or warranty and none of them accept any liability in connection with this paper, including, without limitation, with respect to the completeness, accuracy, or reliability of the information contained herein and the potential legal, compliance, tax, or accounting effects thereof. This document is not intended as investment research or investment advice, or as a recommendation, offer, or solicitation for the purchase or sale of any security, financial instrument, financial product or service, or to be used in any way for evaluating the merits of participating in any transaction.

\clearpage
\bibliographystyle{chicago}
\bibliography{references}

\end{document}